\pdfoutput=1
\documentclass[11pt]{article}
\usepackage[left=1in, top=1in, right=1in, bottom=1in]{geometry}
\usepackage{bookmark}
\usepackage{amsmath, amsthm, amssymb, bm, bbm}
\usepackage{mathtools}
\usepackage{appendix}
\usepackage{color}
\usepackage[semibold]{libertine}
\usepackage{aliascnt}
\usepackage{graphicx}
\usepackage[numbers]{natbib}
\usepackage{enumerate}
\usepackage{thmtools}
\usepackage{thm-restate}
\usepackage{float}
\usepackage{thmtools}
\usepackage{thm-restate}
\usepackage{algorithm2e}

\newtheorem{theorem}{Theorem}[section]
\newtheorem{lemma}[theorem]{Lemma}

\newtheorem{claim}[theorem]{Claim}

\newtheorem{definition}[theorem]{Definition}
\newtheorem{assumption}[theorem]{Assumption}
	\newaliascnt{remark}{theorem}
	
\newtheorem*{remark*}{Remark}
\newtheorem{example}[theorem]{Example}

\RequirePackage{comment}

\newcommand\hufu[1]{}
\newcommand\qunhu[1]{}

\newcommand{\AutoAdjust}[3]{\mathchoice{ \left #1 #2  \right #3}{#1 #2 #3}{#1 #2 #3}{#1 #2 #3} }
\newcommand{\Xcomment}[1]{{}}

\newcommand{\InBrackets}[1]{\AutoAdjust{[}{#1}{]}}% {\left[{#1}\right]}
\newcommand{\Ex}[2][]{\operatorname{\mathbf E}_{#1}\InBrackets{#2}}

\newcommand{\Prx}[2][]{\operatorname{\mathbf{Pr}}_{#1}\InBrackets{#2}}
\newcommand{\given}{\;\vert\;}
\newcommand{\eps}{\epsilon}

\newcommand{\noaccents}[1]{#1}

\newcommand{\newagentvar}[3][\noaccents]{%
\expandafter\newcommand\expandafter{\csname #2\endcsname}{#1{#3}}%
\expandafter\newcommand\expandafter{\csname #2s\endcsname}{#1{\boldsymbol{#3}}}%
\expandafter\newcommand\expandafter{\csname #2smi\endcsname}[1][i]{#1{\boldsymbol{#3}}_{\text{-}##1}}%
\expandafter\newcommand\expandafter{\csname #2i\endcsname}[1][i]{#1{#3}_{##1}}%
\expandafter\newcommand\expandafter{\csname #2ith\endcsname}[1][i]{#1{#3}_{(##1)}}%
}

\DeclareMathOperator{\argmin}{argmin}

\DeclareMathOperator{\reg}{Reg}

\newagentvar{alloc}{x}

\newcommand{\arrival}{\lambda}
\newcommand{\arrivals}{\boldsymbol{\lambda}}
\newcommand{\process}{\mu}
\newcommand{\processes}{\boldsymbol{\mu}}
\newcommand{\paths}{\mathcal P}
\newcommand{\queue}{Q}
\newcommand{\matchingmatrixset}{\mathcal M}
\newcommand{\matchingmatrix}{M}
\newcommand{\probabilitymatrix}{P}
\newcommand{\source}{S_1}
\newcommand{\medium}{S_2}
\newcommand{\terminal}{S_3}
\newcommand{\disjointpathset}{U}
\newcommand{\action}{a}
\newcommand{\potential}{\Phi}

\newcommand{\neighbor}{N}
\newcommand{\outneighbor}{\neighbor^{\textnormal{out}}}
\newcommand{\inneighbor}{\neighbor^{\textnormal{in}}}

\newcommand{\indicator}{\mathbbm{1}}

\newcommand{\history}{\mathcal F}

\newcommand{\alphas}{\boldsymbol{\alpha}}

\newagentvar{age}{T}
\newagentvar{strat}{p}

\title{Stability of Decentralized Queueing Networks \\ {Beyond Complete Bipartite Cases}}

\author{Hu Fu
  \thanks{Shanghai University of Finance and Economics, \href{mailto:fuhu@mail.shufe.edu.cn}{\texttt{fuhu@mail.shufe.edu.cn}}.  The work is supported by the Fundamental Research Funds for the Central Universities of China.}
  \and Qun Hu \thanks{Shanghai University of Finance and Economics, \href{mailto:2019212804@163.shufe.edu.cn}{\texttt{2019212804@163.shufe.edu.cn}}}
  \and Jia'nan Lin \thanks{Rensselaer Polytechnic Institute, \href{mailto:linj21@rpi.edu}{\texttt{linj21@rpi.edu}}.  Part of the work done when the third author was visiting Shanghai University of Finance and Economics.}
}

\begin{document}
\maketitle

\begin{abstract}
  \citet{gaitonde2020stability, gaitonde2021virtues} recently studied a model of queueing networks where queues compete for servers and re-send returned packets in future rounds.  They quantify the amount of additional processing power that guarantees a decentralized system's stability, both when the queues adapt their strategies from round to round using no-regret learning algorithms, and when they are patient and evaluate the utility of a strategy over long periods of time.

  In this paper, we generalize Gaitonde and Tardos's model and consider scenarios where not all servers can serve all queues (i.e., the underlying graph is an incomplete bipartite graphs) and, further, when packets need to go through more than one layer of servers before their completions (i.e., when the underlying graph is a DAG).  For the bipartite case, we obtain bounds comparable to those by Gaitonde and Tardos, with the factor slightly worse in the patient queueing model.  For the more general multi-layer systems, we show that straightforward generalizations of the utility function and servers' priority rules in \citep{gaitonde2020stability} may lead to unbounded gaps between centralized and decentralized systems when the queues use no regret strategies.  We define a new utility and a service priority rule that are aware of the queue lengths, and show that these suffice to restore the bounded gap between centralized and decentralized systems observed in bipartite graphs.
\end{abstract}

\section{Introduction}
\label{sec:intro}

A recurrent theme in algorithmic game theory is to analyze systems operated by decentralized, strategic agents, in comparison with those run by centralized authorities.
Since Koutsoupias and Papadimitriou \cite{KP99} introduced the concept of \emph{Price of Anarchy}, it has been applied and studied in various games such as routing in congestion games \cite{RT02}, network resource allocation \cite{JT04}, auctions \cite{CKS16}, among many other settings.
Recently, Gaitonde and Tardos \cite{gaitonde2020stability,gaitonde2021virtues} introduced a routing game in queueing systems, where queues compete for servers each round, and packets not processed sucessfully in one round go back to their queues and have to be re-sent in the future.
Unlike most games previously studied, in such systems, the strategies and outcomes of one round have carryover effect in future rounds, introducing intricate dependencies among the rounds.
Gaitonde and Tardos developed bicriteria bounds that quantify the loss of efficiency due to decentralized strategic behaviors in such systems in two settings:  in \cite{gaitonde2020stability}, the queues evaluate the utility of their strategies from round to round, and adopt no-regret learning algorithms in their routing decisions; in \cite{gaitonde2021virtues}, the queues are ``patient'', and fix their strategies over long periods of time over which they evaluate their performances.

In both \cite{gaitonde2020stability} and \cite{gaitonde2021virtues}, all servers can process requests from all queues, and a packet leaves the system once it is processed by a server.
These are simplifying modelling assumptions: in many queueing systems, each queue's packets may only be processed by certain servers, and a packet may need to go through more than one server before leaving the system.
In this work, we model such added complexities by seeing the queues and servers as nodes of a directed acyclic graph (DAG).  
A queue can send requests to a server only if it has an outgoing edge to the server.
Packets arrive at given rates to nodes with no incoming edges, and leave the system when they are successfully processed by servers with no outgoing edges; 
nodes with both incoming and outgoing edges are both servers and queues --- after it successfully processes a packet, the packet joins its queue and waits to be sent to the next server.
The case considered by Gaitonde and Tardos \cite{gaitonde2020stability} corresponds to complete bipartite graphs.
We examine whether and how their results generalize to more general settings.

\paragraph{Our Results.}
We first characterize networks that can be stable under a centralized policy, where stability roughly means that the number of packets accumulated in the system is bounded.  
As in \cite{gaitonde2020stability}, the main lesson of the characterization is that it is without loss of generality for a centralized policy to fix for each queue a distribution and sample a server from this distribution at each time step, independently of the history and all other happenings in the system.
For bipartite graphs (Theorem~\ref{thm:IB-central}) our proof takes a perspective  arguably simpler than that in~\cite{gaitonde2020stability}, and this perspective is instrumental in showing the conditions for general DAGs, which are considerably more involved.

We then consider decentralized systems where queues use no-regret learning strategies.
For general bipartite graphs, we show that the bound in \cite{gaitonde2020stability} generalizes with minor modification.  
We inherit much of the proof framework of \cite{gaitonde2020stability}, including a potential function argument and various apparatus for analyzing the random processes, although in the key step of the argument where one uses the no regret property to bound the number of ``old'' packets processed over a time window, our proof has to take into account the underlying graph structure, and makes a connection with the \emph{dual} form of the conditions for centralized stability.
The eventual stability conditions we give (Theorem~\ref{thm:IB-decentral}) when queues use no-regret learning strategies is also expressed as a scaled dual form of the centralized stability conditions. 
As a consequence, the main bicriteria comparison result in \cite{gaitonde2020stability} extends to general bipartite graphs: a decentralized system is stable if it can be made stable under a centralized policy with the arrival rates doubled. 
Interestingly, the dual variables in our decentralized stability condition take values from a smaller range ($\{0, 1\}$) than in the centralized stability condition (where they may be any nonnegative numbers).  
For complete bipartite graphs, it can be shown that even for the centralized stability condition, the dual variables need only take $0, 1$ values.
In this sense, our results suggest that the gap between the two conditions tends to be smaller for incomplete bipartite graphs.

Networks with more than one layer of servers are even more interesting.
A major conclusion reached in~\cite{gaitonde2020stability} is that a server's rule of priority for packets simultaneously sent to it is crucial for the system's stability.
In the complete bipartite graphs, it was shown that, if the servers pick a packet uniformly at random, then no said bicriteria bound can be given; in contrast, the bicriteria result was obtained when servers are assumed to prioritize older packets.
Another important factor in the model is the queues' utilities: it was assumed in~\cite{gaitonde2020stability} that a queue collects utility of 1 if its packet is successfully cleared by a server, and 0 otherwise.
Our results for general bipartite graphs inherit both these modelling assumptions.
However, for graphs of even three layers, we give an example showing that no finite bound of the bicriteria form can be obtained if one directly extends the utility and the priority rule from~\cite{gaitonde2020stability}.
Intuitively, in order for the system not to lose too much efficiency, information on the underlying graph is important when there are multiple layers: a server with strong processing capacity may be poorly connected in the next layer, and myopic strategies easily send too many packets to such a server.
Therefore, the queues' utilities need to incorporate more information for their strategies to better align with the system's stability; on the other hand, if they are fed with too much global information, the difference could blur between centrally controlled systems and decentralized ones.
A natural question to raise is whether it is possible to incentivize the queues using only local information so that their selfish behaviors do not hurt the system efficiency too much.
We answer this question in the affirmative, showing that the \emph{lengths} of queues in the neighboring nodes provide just this information.
We propose a new service priority rule, under which the servers prioritize packets from the \emph{longest} queues.
We also propose new utility functions for queues, with which a queue of length $L_i$, when it sends a packet to a server~$j$ whose own queue is of length~$L_j$, obtains utility $L_i - L_j$ if the packet is successfully processed by~$j$.
In particular, with this new utility function, a queue never sends its packets to a server whose current queue is longer than itself.
We show that when the new service priority rule and utilities are adopted, the bicriteria result is restored: a queueing system is stable with queues that use no-regret strategies as long as it is stable under a centralized policy even when the packet arrival rates are doubled.

Lastly, we extend the model with patient queues to bipartite graphs. Gaitonde and Tardos
\cite{gaitonde2021virtues} showed for complete bipartite graphs that, when queues are patient, with appropriately defined long-term utilities, a Nash equilibrium always exists, and a system is stable under any Nash equilibrium as long as it is stable under a centralized policy even with $\frac e {e - 1}$ times the original arrival rates.
To this end, they developed elaborate tools for computing the long-term utilities given the queues' strategies.  
These tools generalize straightforwardly in general bipartite graphs, but the delicate deformation argument in the proof of their bicriteria result does not easily generalize.
Our proof again makes use of the dual form of the condition for centralized stability, which provides a matching between the fastest growing queues in an equilibrium and servers.

In the appendix, we also consider two other variants of the problem: in one model, whether a server can process a packet is not determined by which queue the packet is from, but is an intrinsic property of the packet; in the other one, the arrival of packets at each queue is not from a Bernoulli distribution, but is controlled by an adversarial, as in the model of  Borodin et al \cite{BKRSW01}.
In both variants, we show that the bicriteria results persist when queues use no-regret strategies.
Lastly, we give a tighter bicriteria result for the model in \cite{gaitonde2020stability}, where the underlying graph is a complete bipartite graph. 
We show that a queueing system is stable with queues that play no-regret strategies as long as it is stable under a centralized policy even when the $k$-th largest packet arrival rate is increased by a factor $\frac{2k-1}{k}$ for each~$k$.

\paragraph{Further Related Works.}
We refer to Gaitonde and Tardos \cite{gaitonde2020stability,gaitonde2021virtues} for pointers to related works in algorithmic game theory and no regret learning.
Sentenac et al.~\cite{sentenac2021decentralized} considered the same model as in~\cite{gaitonde2020stability} but when queues use \emph{cooperative} learning.  When incentives are removed from the problem, they show that the queues can essentially learn the necessary system parameters and reach a stable outcome as long as the system is stable under a centralized policy.

\section{Preliminaries}
\label{sec:prelim}

\subsection{Queue-G Model}
A \emph{Queue-G Model} is a $G = (V, E, \arrivals, \processes)$, where $(V = \source \cup \medium \cup \terminal, E)$ constitutes a directed acyclic graph, and $\arrivals$ and $\processes$ are the arrival and processing rates on the nodes.
A node~$i$ with no incoming edge is a \emph{source}, and has an \emph{arrival rate}~$\arrival_i$. For each $i$, $\arrival_i \in (0,1)$.
$\source$ denotes the set of sources.
All the other nodes are \emph{servers}, and each server~$j$ has a \emph{processing rate}~$\process_j$.
A server with no outgoing edge is a \emph{terminal}.
$\terminal$ denotes the set of terminals.
The set of non-terminal server nodes is $\medium \coloneqq V - \source - \terminal$.
% The node set  $V = \source \cup \medium \cup \terminal$. Each Node $i$ in $\source$ represents a queue and has zero in-degree. Each Node $j$ in $\medium$ represents a server and has nonzero indegree and nonzero outdegree. Each Node $j$ in $\terminal$ represents a server and zero outdegree. 
% There is an edge $(i,j)$ between node $i$ and node $j$ if and only if packets from node $i$ can be served by server $j$.
An edge $(i, j) \in E$ means that node~$i$ can send packets to node~$j$.
For $i \in \source \cup \medium$, we denote by $\outneighbor(i) \coloneqq \{j \in V: (i, j) \in E\}$ the set of out-neighbors of~$i$, and for a server~$i$, we denote by $\inneighbor(i) \coloneqq \{j \in V: (j, i) \in E\}$ the set of in-neighbors of~$i$.

Let $\queue_t^i$ denote the number of packets at node $i$ at the beginning of time step $t$. 
For all $i \in V$, $\queue_0^i = 0$.
% If $\queue_t^i>0$, it means that at the beginning of time step $t$, node $i$ has a packet to be served.
At each time step $t$, the following events happen, in two phases:
\begin{enumerate}[(I)]
  \item Packet sending: 
each node~$i$ with $\queue_t^i>0$ %and has a non-empty queue
chooses a server~$j$ from $\outneighbor(i)$ and sends to $j$ the oldest packet (with the earliest timestamp) in $i$'s queue.
In a centralized system, a central authority dictates for each node if and where to send its packet at each time step;
in a decentralized system, each node strategizes over this decision.
\item Packet arrival and processing: 
at each source $i \in \source$, a packet with timestamp~$t$ arrives with probability~$\arrival_i$;
% , independently of all the other events.
each server $i \in \medium \cup \terminal$, if it receives any packet in phase (I), chooses one such packet according to some \emph{service priority rule} to process, and succeeds with probability~$\process_i$.
% it chooses the oldest packet (breaking ties arbitrarily) to proceed, and succeeds with probability~$\process_i$.
% \hufu{Instead, should define serving priority here.}
% a packet is cleared successfully by server~$j$ with probability~$\process_j$, independently of all the other events. 
The arrivals of packets at each source and the successes of their processing at each server are all mutually independent events.
A packet cleared by server $j \in \medium$ joins the queue of server~$j$; a packet cleared by a server in~$\terminal$ leaves the system. 
A packet not chosen by or not successfully processed by a server goes back to the node that sends it.
It follows that any $i \in \terminal$ has $\queue_t^i=0$ at any time step~$t$. 
\end{enumerate}
% If $\medium = \emptyset$, then the general graph is reduced to a bipartite graph.  

Gaitonde and Tardos~\cite{gaitonde2020stability} considered a special case of the Queue-G Model, where there are no non-terminal servers and every source can send packets to every server, i.e., $\medium = \emptyset$ and $E = \source \times \terminal$, and the service priority rule at each server is to choose the oldest packet (breaking ties arbitrarily).\footnote{For ease of presentation, we made minor changes from \citet{gaitonde2020stability}'s model, in the order of packet sending and packet arrival. It is easy to see that the difference is negligible for the analysis of the system's stability, which is an asymptotic quality, to be defined below.}
 %  we restrict that $\medium$ is an empty set and the corresponding graph is a complete bipartite graph. 
 We refer to this special case as
 % the model in \cite{gaitonde2020stability}
 the \emph{Queue-CB Model} (``CB'' for complete bipartite). 
% If we generalize \textbf{Queue-CB Model} a little bit, we have a new model called 
If we only have $\medium = \emptyset$ (and allow any $E \subseteq \source \times \terminal$), we have the \emph{Queue-B Model}. 

\subsection{Stability and no regret learning}\label{stability and no regret}
Let $\queue_t \coloneqq \sum_{i \in V}\queue_t^i$ be the total number of packets in the queueing system at the start of time step~$t$.
We inherit from \cite{gaitonde2020stability} the notion of stability:

\begin{definition}% [Strongly stable]
  \label{def:stability}
  Under some scheduling policy (either with a central authority or with queues strategizing), a queueing system is \emph{strongly stable} if for any $a>0$, there is a constant $C_a$ only related to $a$, such that $\mathbb{E}[(\queue_t)^a] \leq C_a$ for all $t$. A queueing system is \textbf{almost surely strongly stable} if with probability 1, the following event happens: for any $a>0$, $\queue_t=o(t^a)$.
\end{definition}
\citet{gaitonde2020stability} showed that if a queueing system is strongly stable, then it is almost surely strongly stable. 
We therefore focus on showing strong stability, and often refer to a strongly stable system simply as stable.
% \hufu{Rewrite this definition.  Introduce almost surely strongly stable and cite the equivalence between the two.}

% The main mathematical tool we use to prove the theorem is a well known theorem by Pemantle and Rosentale~\cite{pemantle1999moment}:
The following theorem by \citet{pemantle1999moment}, also used in \citep{gaitonde2020stability}, is the workhorse for showing stability.

\begin{theorem}[\citeauthor{pemantle1999moment}] \label{pemantale}
  \label{thm:pemantle}
  Let $X_0, \cdots, X_n$ be nonnegative random variables. 
  If there are constants $b,c,d>0$ and $p>2$ such that $X_0 \leq b$ and, for all $n$, 
\begin{align}
  \mathbb{E}(|X_{n+1}-X_n|^p \given X_0, \cdots, X_n) & \leq d; \label{eq:bounded-jump} \\
  X_n>b \quad \Rightarrow  \quad \mathbb{E}(X_{n+1}-X_n \given X_0, \cdots, X_n) & \leq -c,  \label{eq:negative-drift}
\end{align}
then for any $a \in(0,p-1)$, there is $C=C(p,a,b,c,d)$ such that $\mathbb{E}(X_n)^a <C$ for all $n$.
\end{theorem}

We refer to \eqref{eq:negative-drift} as the \emph{negative drift condition}, and \eqref{eq:bounded-jump} as the \emph{bounded jump condition}.

We now introduce utilities of queues, as defined in~\citep{gaitonde2020stability}. The utility of a queue at time step~$t$ is the number of packets cleared from this queue at time step~$t$.
Let $a_i(t)$ denote the server that node~$i$ chooses at time step~$t$.  
(A node $i$ may not choose any server, in which case we let $a_i(t)=0$ and we set $\process_0=0$.)
Let $\mathcal{F}_{t}$ denote the history of the system up to the beginning of time step~$t$. 
We use $u_t^i(a_i(t),a_{-i}(t))|\mathcal{F}_{t})$ to denote the utility of node $i$ when node $i$ chooses server $a_i(t)$ and the other nodes choose $a_{-i}(t)$, given history~$\mathcal{F}_{t}$. 
We should specify the content of a history: $\mathcal{F}_{t}$ only includes information on which packets, up to time step~$t$, were cleared and the age of the currently oldest packet in each node, but does not include the queue size $Q_t^i$. 
This makes sure that, for the $k$-th packet in node $i$ that is cleared at time step~$t$, the time difference between its arrival and that of the $(k+1)$-st packet is independent of the history $\mathcal{F}_{t'}$ for all $t' <t$, and obeys the geometric distribution with parameter $\arrival_i$.
% In addition, we require that the history 

% In the Queue-CB and the Queue-B model, the \emph{utility} of node~$i$ of time step~$t$, $u_t^i(a_i(t),a_{-i}(t)|\mathcal{F}_{t})$, is defined to be the number of packets sent by~$i$ and cleared during that step.  
%This number we denote as $x_t^i$.
%\hufu{I don't see where we use this notation.  If not, we shouldn't introduce the notation.}
% The notation for utility underlines that it is a function of both node $i$'s action and the other nodes' action at time step~$t$, and is also dependent on the history up to that point in time.
% leaving node~$i$.  
% Let $x_t^i$ denote the number of packets leaving node~$i$ at time step~$t$. 
% Then, 
% In \textbf{Queue-GN Model},$u_t^i(a_i(t),a_{-i}(t)|\mathcal{F}_{t})=x_t^i(Q_t^i-Q_t^{a_i(t)})$. 
Lastly, we define the \emph{regret} of a node~$i$ up to time~$w$ as the difference between its utility in a real sample path and what it could have achieved by always playing a best fixed action in hindsight.  
% during a time window of length $w$.
\begin{definition}
  \label{def:regret}
  % For any starting time step $t_0$,  we denote the \emph{regret} of queue~$i$ during time $[t_0, t_0+w]$ as $\reg_i(w, t_0)$.
    % Given a time window of length $w$, we denote the \textbf{regret} of queue $i$ by $\reg_i(w)$, 
For a time window from time step $t_0 - w$ to $t_0-1$, the \emph{regret} of queue~$i$ for actions $a_i(t_0 - w), \ldots, a_i(t_0-1)$ is
\begin{align*}
  \reg_i(w, t_0)\coloneqq \max \limits_{j:j\in \outneighbor(i)} \sum_{t=t_0 - w}^{t_0-1} u_t^i(j, \action_{-i}(t)|\mathcal{F}_{t}) - \sum_{t=t_0 - w}^{t_0-1} u_t^i(\action_i(t), \action_{-i}(t) |\mathcal{F}_{t}).
\end{align*}
\end{definition}

Note that, in this definition, the utility obtained by playing the best fixed strategy is evaluated using ``real'' histories ($\mathcal F_t$'s) observed under the actual actions taken by the node.
It does not use counterfactual histories generated by playing the fixed strategy.
We often drop the parameter $t_0$ when it is clear from the context.  

% Note that the compared utility is sum of fictitious utility during each time step~$t$ if node $i$ chooses server $j$, given real history $\mathcal{F}_t$.
\begin{definition}
  \label{def:no-regret}
  Given fixed $\delta \in (0,1)$, queue~$i$'s scheduling policy is \emph{no regret} if, for any time window from time step $t_0 - w$ to~$t_0-1$, with probability at least $1-\delta$, 
$\reg_i(w, t_0)\leq \varphi_{\delta}(w)$,
where $\varphi_{\delta}(w) = o(w)$ may depend only on $\delta$ and the number of nodes in the queueing system.
\end{definition}

\section{Bipartite Queueing Systems}
\label{sec:queue-b}
\label{sec:IB}

In this section we derive necessary and sufficient conditions for the existence of a centralized policy that stabilizes a queueing system in a bipartite graph (a Queue-B model). 
We then give a sufficient condition that guarantees the stability of such systems when all queues adopt no-regret strategies.  
For the special case when the underlying graph is complete bipartite, our conditions degenerate to the ones given by \citet{gaitonde2020stability}.

\subsection{Stability Conditions under Centralized Policies}
\label{sec:IB-central}

% \begin{definition}(Matching matrix)

A Queue-B model as defined in Section~\ref{sec:prelim} simply consists of $n$ queues on one side and $m$ servers on the other.  
Server~$j$ is able to clear a packet from queue~$i$ if and only if there is an edge between the two.  
It is easy to see that a centralized policy never benefits from sending packets from two queues to a same server in a single time step, as the server picks up only one of them.
Therefore, with loss of generality, the routing dictated by a centralized policy at any step gives a matching of the queues to the servers.  (Some queues may be asked not to send their packets, and some servers may be allowed to be idle for that round.)
It is less clear whether a centralized policy benefits from making intricate use of the history when it decides on the matching at each step.
It turns out, for the system to be stable (Definition~\ref{def:stability}), it is without loss of generality to consider history oblivious centralized policy, which samples a matching from a fixed distribution over matchings from step to step.
A Queue-B model can be stable under any centralized policy if and only if it can be stable under such a policy.
This is the essence of the following theorem.

% Recall that a $\matchingmatrix$ is an $n \times m$ matrix whose elements are 0 and 1. It satisfies the following condition: $\forall i,j, \sum_i M_{ij}\leq 1,\sum_j M_{ij}\leq 1, M_{ij}=0$ if server j can not serve packets from queue $i$. We assign server $j$ to queue $i$ if $M_{ij}=1$. $\matchingmatrixset$ represents the set of all matching matrices for queues and servers. Denote $\mathbf{x}$ a probability vector over $\mathcal{M}$. Given $\mathbf{x}$, we can construct a probability matching matrix $P:=\sum_{k}x_{k}M_{k}$.

% Next, we give a sufficient and necessary stability condition of centralized \textbf{Queue-IB} Model.

Recall that a fractional matching matrix $P \in [0, 1]^{n \times m}$ is such that $\sum_j P_{ij} \leq 1$ for all $i \in [n]$ and $\sum_i P_{ij} \leq 1$ for all $j \in [m]$.
\begin{restatable}{theorem}{IBcentral}
\label{centralize IB}
  \label{thm:queue-b-centralized}
  \label{thm:IB-central}
  Given a Queue-B model with $n$ queues and $m$~servers, with arrival rates $\arrivals = (\arrival_1, \cdots, \arrival_n)$ and processing rates $\processes = (\process_1, \ldots, \process_m)$, there is a centralized policy under which the system is stable if and only if there exists a fractional matching matrix~$P \in [0, 1]^{n \times m}$, such that $P \processes \succ \arrivals$, where $\succ$ denotes element-wise greater than.  
\end{restatable}

Sufficiency of the condition is a consequence of Birkhoff-von Neumann theorem.  
The argument of necessity makes use of the observation that, conditioning on any event in the system, the expected routing decision made by a centralized policy is expressible as a fractional matching matrix.  
One may as well condition on the event that all queues have arrivals considerably larger than the expectations, which occurs with constant probability.
This part of the argument is arguably simpler than the proof in \citep{gaitonde2020stability}, and makes possible the more involved proof for more general graphs (Theorem~\ref{thm:G-central}).
We relegate the details to Appendix~\ref{sec:app-IB-central}. 

Before moving on to decentralized Queue-B models, we derive a dual form of the conditions in Theorem~\ref{thm:queue-b-centralized}.
The dual form plays a crucial role in our analysis of the systems' stability under no-regret policies.  

\begin{restatable}{lemma}{IBdual}
%\label{dual condition IB}
  \label{lem:IB-dual}
  Given a Queue-B model with arrival rates $\arrivals$ and processing rates~$\processes$, the following two conditions are equivalent: % by Farkas' lemma.\\
\begin{enumerate}[(1)]

  \item There is a fractional matching matrix $\probabilitymatrix$ such that $P\boldsymbol{\process} \succ \boldsymbol{\arrival}$.

  \item For any $\boldsymbol{\alpha} \in \mathbb R_+^n$, there is a matching matrix $M \in \{0, 1\}^{n \times m}$, such that 
    $\boldsymbol{\alpha}^\top M \processes > \boldsymbol{\alpha}^\top \arrivals$.
    % $\sum_{i=1}^{n}\alpha_i\process_{M(i)}>\sum_{i=1}^{n}\alpha_i \lambda_i$, where $M(i)=j$ if $M_{ij}=1 $ and $\process_{M(i)}=0$ if $M_{ij}=0$ for all server $j$.\\
\end{enumerate}
\end{restatable}
% \hufu{Rewrite the second condition using sums.}

The lemma is an application of Farkas' lemma.  We relegate its proof to Appendix~\ref{sec:app-IB-central}. 
It is worth pointing out that, when the underlying graph is a complete bipartite graph, it suffices to have the condition (2) satisfied for only $\boldsymbol \alpha \in \{0, 1\}^n$.
This difference plays a role in the contrast between complete and incomplete bipartite graphs when the system is decentralized, as we explain in the next section.

% \subsection{Decentralized stability condition}
\subsection{Stability Conditions under Decentralized, No-Regret Policies}
\label{sec:IB-decentral}

In this section we give conditions under which, in a queueing system on an incomplete bipartite graph (the Queue-B model), if all queues use no-regret strategies, the system is stable.
Our conditions are most easily comparable with the dual form of centralized stability conditions stated in Lemma~\ref{lem:IB-dual}.  
When the underlying graph is a complete bipartite graph, the conditions are identical to those by \citet{gaitonde2020stability}, as we discuss below.
The technique in this part is largely inherited from \citep{gaitonde2020stability}, although our proof reveals an interesting connection between the dual form of stability conditions and key steps in the proof.
The sufficient condition is the following:
\begin{assumption}\label{decentralized IB stability condition}
  \label{assump:IB}
  There is a constant $\beta>0$, such that for any $\boldsymbol{\alpha}=(\alpha_1,...,\alpha_n)\in \{0,1\}^n$, there is a matching matrix $M$, such that $\frac 1 2 (1 - \beta) \boldsymbol{\alpha}^\top M \processes > \boldsymbol{\alpha}^\top \arrivals$.
  % $\frac{1}{2}(1-\beta)\sum_{i=1}^{n}\alpha_i\process_{M(i)}>\sum_{i=1}^{n}\alpha_i \lambda_i$
\end{assumption}

A quick comparison between this and dual condition in Lemma~\ref{lem:IB-dual} suggests that, if one has a Queue-B model which can be made stable by a centralized policy, then, doubling its processing capabilities guarantees its stability when the queues use no-regret strategies.
Note though that the range of $\alphas$ is much smaller in Assumption~\ref{assump:IB} ($\{0, 1\}^n$) than in Lemma~\ref{lem:IB-dual} ($\mathbb R_+^n$). 
For complete bipartite graphs, this difference vanishes (see remark following Lemma~\ref{lem:IB-dual}), but in general bipartite graphs, this difference is real.
(See Example~\ref{ex:fractional-dual}.)
This suggests that in incomplete bipartite graphs, the gap between centralized and decentralized systems tends to be smaller than in complete bipartite graphs.
% which means that in general, given a Queue-B model satisfying Lemma~\ref{lem:IB-dual}, one only needs to increase its processing capabilities by a factor smaller than~2 to satisfy Assumption~\ref{assump:IB}.  
% \hufu{Give our small example in the appendix.}

% Now to formally state the result, we need to fix a few parameters in the definition of no-regret strategies.
\begin{restatable}{theorem}{thmIBdecentral}\label{decentralized IB}
  \label{thm:IB-decentral}
  If a Queue-B model queueing system satisfies Assumption \ref{decentralized IB stability condition}, and queues use no-regret learning strategies with 
  $\delta=\frac{\beta}{128n}$, then the system is strongly stable.
\end{restatable}
% \hufu{I removed ``with large enough~$w$''.  The definition of no regret seems to have no parameter~$w$.}\qunhu{definition of no regret has $w$ and $\delta$}

% \begin{proof}
  Following \citet{gaitonde2020stability}, we introduce a potential function with the intention to apply Theorem~\ref{thm:pemantle} to its square root.
  The \emph{age} of a packet that arrives in the system at time $t_1$ is defined to be $t_2-t_1$ at time $t_2$.
  Let $T_t^i$ be the age of the oldest packet in queue $i$ at time step~$t$, and let $\ages$ be the vector $(\age_t^1, \cdots, \age_t^n)$.
  Note that $Q_t^i$, the length of the queue, is at most $T_t^i$.  
  For a positive integer $\tau > 0$, define 
\begin{align*}
  \potential_{\tau}\left(\mathbf{T}_{t}\right)\coloneqq \sum_{i:T_t^i \geq \tau} \arrival_i(T_t^i-\tau).
\end{align*}
The potential function~$\potential$ is defined as
\begin{align*}
  \potential \left(\mathbf{T}_{t}\right) \coloneqq \sum_{\tau=1}^{\infty} \potential_{\tau}\left(\mathbf{T}_{t}\right)= \sum_{\tau=1}^{\infty} \sum_{i:T_t^i \geq \tau} \arrival_i(T_t^i-\tau) = \frac{1}{2} \sum_{i=1}^{n}  \arrival_i T_t^i(T_t^i-1).
\end{align*}

We analyze the system by dividing the time steps into windows of length~$w$ each, for some large enough~$w$.
Let $Z_\ell \coloneqq \sqrt{\potential\left(\mathbf{T}_{\ell \cdot w} \right)}$ be the square root of the potential function at the beginning of the $\ell$-th window.
The main work lies in showing that $(Z_{\ell})_\ell$ satisfies the conditions of Theorem~\ref{thm:pemantle}, which implies $\Ex{Z_\ell^a}$ is bounded for any $a > 0$.
This in turn implies that $\Ex{(\sum_i T_t^i)^a}$ is bounded, and so is $\Ex{(Q_t)^a}$.  
% Therefore showing that $(Z_\ell)_\ell$ satisfies the negative drift condition and the bounded jump condition in Theorem~\ref{thm:pemantle} amounts to showing that the system is stable.

\begin{restatable}{lemma}{negativedrift}[Negative drift condition.]\label{threshold value 1}
  \label{lem:negative-drift}
Denote by $\lambda_{(n)}$ the minimum element of $\arrivals$. \\
Let $b=\frac{w}{\sqrt{2 \lambda_{(n)}}} \max \left(\frac{8}{\beta}\left(\sum_{i=1}^{n} \lambda_{i}\right), 16 n^{2}\right)$, $c=-\frac{\sqrt{2\arrival_{(n)}}\beta w}{64}$. Then $Z_0=0 \leq b$ and, for all $\ell$,
\begin{align*}
    Z_{\ell}>b \quad \Rightarrow  \Ex{Z_{\ell+1}- Z_{\ell} \given Z_0, \cdots, Z_{\ell}} \leq -c.
\end{align*}

\end{restatable}

\begin{restatable}{lemma}{boundedjump}[Bounded jump condition.]
  \label{lem:bounded-jump}
  For each even integer $p \geq 2$, there is a
  constant $d_p$, such that for all $\ell$
  \begin{align*}
  \Ex{|Z_{\ell+1}-Z_{\ell}|^p \given Z_0, \cdots, Z_{\ell}} \leq d_p.
  \end{align*}
\end{restatable}

Lemma~\ref{lem:bounded-jump} is identical to the corresponding part
% the \emph{'Bounded $p$th Moments'} paragraph at the end of proof of Theorem 3.1 
in~\citet{gaitonde2020stability}, and we omit its proof.  
% \hufu{give the number here}\qunhu{There is no explicit lemma in their paper.Lemma~??}
The main difference between our proof and \citep{gaitonde2020stability} is in the proof of the negative drift condition (Lemma~\ref{lem:negative-drift}).  
We present here the key steps of our proof, and relegate the rest to Appendix~\ref{sec:app-IB-decentral}.  

Following \citet{gaitonde2020stability}, for a given $\tau > 0$, we say a packet is \emph{$\tau$-old} if its age is at least $\tau$ at time step $\ell \cdot w$, i.e., if its arrival time is no later than $\ell w - \tau$.
Let $J_{\tau}$ be the set of queues which have $\tau$-old packets at time step $(\ell+1)\cdot w$. 
For a queue~$i$, if by time step~$(\ell + 1) \cdot w$, it still has packets that arrived before time step $\ell \cdot w$, let $\tau_i = \max_{\tau > 0: J_{\tau} \ni i} \tau$ be the age of the oldest packet in queue $i$;
otherwise, 
%all packets arriving before time step $\ell \cdot w$ are cleared in queue $i$ by time step $(\ell + 1)\cdot w$, then 
set $\tau_i=0$.  
Let $N_{\tau}^i$ be the number of $\tau$-old packets cleared from queue $i$ during the time window from time step $\ell \cdot w$ to $(\ell + 1) \cdot w$. 
% Let $N_{\tau}^i$ be the number of $\tau$-old cleared from queue $i$ during the time interval. 
Similarly, for a server~$j$, let $L_{\tau}^j$ be the number of $\tau$-old packets cleared by server~$j$ during this time window.
 %  Let $N_{\tau}$ denote the number of $\tau$-old packets cleared during the time interval. Then, 
Next, define $N_{\tau}=\sum_{i\in [n]}N_{\tau}^i=\sum_{j\in [m]}L_{\tau}^j$ as the number of $\tau$-old packets cleared during this time window.
Lastly, let $C_t^j$ be the indicator variable for server~$j$ succeeding in processing a packet if it picks one up.
\begin{lemma}
  \label{lem:lower-bound-Ntau}
  For any $\tau > 0$ and $\eps > 0$, if \ $\sum_{t = \ell \cdot w}^{(\ell + 1)\cdot w - 1} C_t^j \geq (1 - \eps) \process_j w$ for each~$j$, then $N_{\tau} \geq \frac {1 - \eps}{1 - \beta} \sum_{i \in J_\tau} \arrival_i w - \sum_{i = 1}^n \reg_i(w, (\ell+1) \cdot w)$.
\end{lemma}

% Given any $\tau>0$, we give a lower bound on $N_{\tau}$. 
\begin{proof}
Any queue $i \in  J_{\tau}$ has a $\tau$-old packet throughout the time window.
% , it always could send packets during the time interval. 
For a server~$j$ which can serve queue~$i$, consider the counterfactual utility $i$ may gain during this time window by sending a request to~$j$ at each step.
Let $X_{jt}^{\tau}$ be the indicator variable for the event that \emph{some queue} (which may not be $i$) sends a $\tau$-old packet to server $j$ at time step $t$.
Then at any time step~$t$ when $X_{jt}^\tau = 0$, server~$i$'s packet would have been picked up by server~$j$ had $i$ sent a request, because no other packet sent to~$j$ is $\tau$-old, so the packet from~$i$ has priority.
Recall that $C_t^j$ is the indicator variable for server~$j$ succeeding in processing a packet if it picks one up.
So queue~$i$ would have gained utility $1$ at time~$t$ by sending a request to~$j$ if $C_t^j = 1$ and $X_{jt}^\tau = 0$.
Over the time window, queue~$i$'s counterfactual utility could have been $\sum_{t = \ell \cdot w}^{(\ell + 1)\cdot w - 1}(1-X_{jt}^{\tau})$.
Note that queue $i$'s actual utility is $N_\tau^i$, so by definition of regret, we have
% Since queue~$i$ uses no-regret learning strategy, we have:
\begin{align*}
    N_{\tau}^i\geq \sum_{t = \ell \cdot w}^{(\ell + 1)\cdot w - 1} C_t^j(1-X_{jt}^{\tau})-\reg_i(w, (\ell+1) \cdot w).
\end{align*}
On the other hand, whenever $X_{jt}^{\tau}C_{t}^{j}=1$, server~$j$ successfully clears a $\tau$-old packet.
% it means that server $j$ clears a $\tau$-old packet. 
Therefore, $L_{\tau}^j=\sum_{t = \ell \cdot w}^{(\ell + 1)\cdot w - 1} C_t^jX_{jt}^{\tau}$.
Then, for a pair of queue $i \in J_\tau$ and server $j$ that can serve~$i$, we have
\begin{align}\label{eq: lower bound of a pair}
    N_{\tau}^i+L_{\tau}^j \geq \sum_{t = \ell \cdot w}^{(\ell + 1)\cdot w - 1} C_t^j- \reg_i(w,(\ell+1) \cdot w).
\end{align}
Now we are ready to apply Assumption~\ref{assump:IB}.
Let $\alphas$ be the indicator vector for the set $J_\tau \subseteq [n]$, i.e., 
% For a set $J_{\tau} \subseteq [n]$, let 
$\alpha_i=1$ if $i \in J_{\tau}$, and $\alpha_i = 0$ otherwise. 
By Assumption \ref{decentralized IB stability condition}, we can find a matching matrix $\matchingmatrix_{\tau}$ such that $\frac{1}{2} (1 - \beta) \boldsymbol{\alpha}^\top \matchingmatrix_{\tau} \processes > \boldsymbol{\alpha}^\top \arrivals$. 
Let $\disjointpathset_{\tau}$ be the edge set such that $(i, j) \in U_\tau \Leftrightarrow \matchingmatrix_{\tau}(i,j)=1$. Then,
\begin{align}\label{eq: decentralized-IB condition}
\frac{1}{2} (1 - \beta)\sum_{(i,j)\in \disjointpathset_{\tau}} \process_j > \sum_{ i \in J_{\tau}} \arrival_i.
\end{align} Now, we are ready to give a lower bound for $N_{\tau}$:
\begin{align}
2N_{\tau} &= \sum_{i=1}^n N_{\tau}^i+\sum_{j=1}^m L_{\tau}^j \geq \sum_{(i,j) \in \disjointpathset_{\tau}}(N_{\tau}^i+ L_{\tau}^j) \geq \sum_{(i,j) \in \disjointpathset_{\tau}} \left( \sum_{t = \ell \cdot w}^{(\ell + 1)\cdot w - 1} C_t^j- \reg_i(w, (\ell+1) \cdot w) \right) \nonumber \\
& \geq \sum_{(i,j) \in \disjointpathset_{\tau}}(1-\epsilon)\process_jw- \sum_{i=1}^n \reg_i(w, (\ell+1) \cdot w) \geq \frac{2(1-\epsilon)}{1-\beta} \sum_{ i \in J_{\tau}} \arrival_i w-\sum_{i=1}^n \reg_i(w, (\ell+1) \cdot w), \label{eq: lower bound of N_tau}
\end{align}
where the second inequality uses \eqref{eq: lower bound of a pair}, 
% the third inequality uses \eqref{eq:concentrate on process 1} under \emph{event A} 
and the last inequality uses \eqref{eq: decentralized-IB condition}.
\end{proof}

We sketch the rest of the proof, and relegate all details to Appendix~\ref{sec:app-IB-decentral}.
When the servers' realized processing capacities are close to their expectations (as in the condition of Lemma~\ref{lem:lower-bound-Ntau}) and when the queues' regret are small (which should happen with high probability by assumption), the lower bound given by Lemma~\ref{lem:lower-bound-Ntau} on $N_\tau$ implies a lower bound on the decrease in the potential function due to packet clearing (Lemma~\ref{lem:phi-decrease}).
% When the packet arrivals in each queue are close to their expectations, we can further bound the increase in the potential function due to packet arrival (Lemma~\ref{lem:phi-increase}).
We can further bound the increase in the potential function due to aging over the time interval. (Lemma~\ref{lem:phi-increase}).
We define an event~$A$ (Definition~\ref{def:eventA}), which happens with high probability (Claim~\ref{claim:eventA}), and under which all of these events (of concentration and no regret) happen.

Recall that $\tau_i$ is the age of the oldest packet in queue $i$ at time step $(\ell+1)\cdot w$, where the age is measured by time step $\ell \cdot w$. Let $\bm{\tau}=\{\tau_1, \cdots, \tau_n\}$.
  \begin{restatable}{lemma}{phidecrease}
    \label{lem:phi-decrease}
  Under \emph{event A}, $\Phi(\mathbf{T_{\ell \cdot w}})-\Phi(\bm{\tau}) \geq \frac{1-2\epsilon}{1-\beta}\sum_{i=1}^n \arrival_i \tau_i w$.
\end{restatable}

\begin{restatable}{lemma}{phiincrease}
  \label{lem:phi-increase}
  Under \emph{event A}, $\Phi(\mathbf{T_{(\ell+1) \cdot w}})-\Phi(\bm{\tau}) \leq \sum_{i=1}^n  \lambda_{i} \tau_i w+\frac{1}{2} \sum_{i=1}^n \lambda_{i}w^{2}$.
\end{restatable}

With a small probability, event~$A$ does not happen, and it is relatively straightforward to upper bound the increase in the potential in this case. (Most pessimistically, no packets is cleared during the time window and $T_t^i$ in each queue grows by~$w$.)
%(Most pessimistically, the length of each queue grows by~$w$.)

\begin{restatable}{lemma}{notA}
  \label{lem:not-A}
  If event~$A$ does not happen, $\Phi(\mathbf{T_{(\ell+1) \cdot w}})-\Phi(\mathbf{T_{\ell \cdot w}}) \leq \sum_{i=1}^n  \lambda_{i}T_{\ell\cdot w}^i w+\frac{1}{2}\sum_{i=1}^n\lambda_{i}w^{2} $.
\end{restatable}

Lemma~\ref{lem:negative-drift} follows from combining Lemma \ref{lem:phi-decrease}, \ref{lem:phi-increase} and~\ref{lem:not-A}.

\section{Queueing Systems with Multiple Layers}
\label{sec:G}

In this section we study queueing systems where packets or tasks may need to go through more than one servers before their completions.  
After a packet is successfully processed by an intermediate server, it immediately joins the queue forming at their server, waiting to be sent to the next server.
In Section~\ref{sec:G-central}, we give sufficient and necessary conditions for such a queueing system to be stable under a centralized policy.  
In Section~\ref{sec:G-example}, we show that, when one extends the utility and service priority rules from \citet{gaitonde2020stability}'s model to such networks, it is impossible to obtain a PoA result comparable to Theorem~\ref{thm:IB-decentral}.
In Section~\ref{sec:G-decentral-sol}, we introduce new utilities and service priority rules that are aware of local queue lengths, and show that they suffice to restore conditions for stability under decentralized, no-regret strategies.

\subsection{Stability under Centralized Policies}
\label{sec:G-central}

As we reasoned for the bipartite case, it never benefits a central planner to send packets from more than one queues to the same server in a single time step, therefore, it is without loss of generality to consider policies under which, at each time step, the edges along which packets are sent from a set of vertex-disjoint paths. 
(Note that one such path need not start from a source or end at a terminal.)
In general, at each step this set of paths may be sampled from a distribution that depends on the history.
As in the bipartite case, the following characterization of stable systems shows it without loss of generality to let this distribution be the same from step to step, regardless of what happened in the past.
The proof though is considerably more involved than in the bipartite case.\footnote{It is relatively easy to extend the argument in Theorem~\ref{thm:IB-central} to show the necessity of the conditions in Theorem~\ref{thm:G-central}, except for the strictness of the signs in \eqref{eq:source-process1} and \eqref{eq:middle-process1}.}
\hufu{Say something about the alternative conditions.}

\begin{restatable}{theorem}{Gcentral}
  \label{thm:G-central}\label{central network}
  Given a Queue-G model $(V, E, \arrivals, \processes)$, the following statements are equivalent.

  \begin{enumerate}
    \item There exists a centralized policy under which the system is stable.

    \item The following linear system is feasible:
\begin{align}
\arrival_i &< \sum_{j} z_{ij} \process_j, \quad & \forall i \in \source;
\label{eq:source-process1}
\\
\process_i \sum_j z_{ji}  & < \sum_j z_{ij} \process_j, & \forall i \in \medium \textnormal{ with } \prod_{j \in \inneighbor(i)} z_{ji} > 0;
\label{eq:middle-process1}
\\
\sum_j z_{ij} & \leq 1, & \forall i \in \source \cup \medium;
\label{eq:flow-outgoing1}
\\
\sum_j z_{ji} & \leq 1, & \forall i \in \medium \cup \terminal;
\label{eq:flow-incoming1}
\\
z_{ij}& =0, &\forall (i,j)\notin E; 
\\
z_{ij} & \geq 0, & \forall (i, j) \in E.
\label{eq:zij-nonnegative}
\end{align}

\item The following linear system in $(f_{i\pi})_{i \in \source, \pi \in \Pi}$ is feasible, where $\Pi$ is the set of paths from a node in~$\source$ to a node in~$\terminal$:
\begin{align}
  \sum_{i \in \source} \sum_{y \in \outneighbor(x)} \sum_{\pi \in \Pi: \pi \ni (x, y)} \frac{f_{i\pi}}{\process_y} & \leq 1, & \forall x \in \source \cup \medium; 
  \label{eq:flow-source}
  \\
  \sum_{i \in \source} \sum_{y \in \inneighbor(x)}\sum_{\pi \in \Pi: \pi \ni (y, x)} f_{i\pi} & \leq \process_x, & \forall x \in \medium \cup \terminal; 
  \label{eq:flow-incoming}
  \\
  \sum_{\pi \in \Pi} f_{i \pi} & > \arrival_i, & \forall i \in \source; 
  \label{eq:flow-process}
  \\
\sum_{i \in \source} \sum_{\pi \in \Pi: \exists y, (y, x) \in \pi} f_{i \pi} - 
\sum_{i \in \source} \sum_{y \in \outneighbor(x)}\sum_{\pi \in \Pi: \pi \ni (x, y)} f_{i\pi} &= 0, & \forall x \in \medium;
\label{eq:flow-conservation}
\\
    % \sum_{y}\sum_{i} \sum_{(x,y)\in p_j} f_{ij}/\process_y    \le 1 ,  & \forall x \in \source \cup \medium \\
    % \sum_{i}\sum_{j: x \ on \ p_j}{f_{ij}} \le \process_x,  & \forall x \in \medium \cup \terminal \\
    % \sum_j f_{ij} > \arrival_i, & \forall i \in \source\\
  f_{i\pi}&=0, & \forall i \in \source, \pi \in \Pi \, \text{s.t. } i \textnormal{ not on } \pi.
  \label{eq:flow-zero}
\end{align}
  \end{enumerate}
\end{restatable}

We relegate the proof of the theorem to Appendix~\ref{sec:app-G-central}.  
The fact that constraints \eqref{eq:source-process1}-\eqref{eq:zij-nonnegative} being feasible implies the stability of a centralized policy is a relatively straightforward consequence of a generalization of Birkhoff-von Neumann theorem (Lemma \ref{lem:path-decomp} and~\ref{lem:2to1}).
Proving the other direction is considerably more involved than for the bipartite case, and it is for this purpose that we introduce the third condition in Theorem~\ref{thm:G-central}.
We show that the feasibility of constraints \eqref{eq:flow-source}-\eqref{eq:flow-zero} implies the feasibility of constraints \eqref{eq:source-process1}-\eqref{eq:zij-nonnegative} (Lemma~\ref{lem:3to2}), then we show that a system for which constraints \eqref{eq:flow-source}-\eqref{eq:flow-zero} are not feasible cannot be stable (Lemma~\ref{lem:1to3}).

% \paragraph{Main idea  of Theorem \ref{thm:G-central}}
% Proving sufficient conditions is easy, when the linear system is feasible, queue i chooses server j with probability $z_{ij}$. We prove that $Q_t^i$ satisfies Theorem \ref{thm:pemantle}. Under this policy, for each node $i$, the probability a new packet joins the queue of the node is strictly smaller than the probability a packet from that node is successfully processed by a next server. When $Q_t^i>0$, negative drift satisfies and $Q_t^i$ changes at most 1 during each time step, satisfying bounded jump condition.
% For necessary conditions, we provide another equivalent linear system which is represented by flow of an entire path $f_{i\pi}$. We define the event that each queue $i$’s arrivals is larger than its expectation $\lambda_i T + \sqrt{T}$, the probability the event happens is a constant and this event is independent of servers’ processing. When the linear system is not feasible, under this event, there is a queue $i$ such that expectation of flow of paths starting from  queue $i$ is no more than $\lambda_i T$, causing packets accumulation. 

% Before moving on to decentralized Queue-G models, 
Again we give a dual form of conditions for centralized stability,
% conditions in Theorem~\ref{thm:G-central}.
% It turns out that the dual form plays a central role 
which are central to our analysis of the systems' stability under no-regret policies.  
Its proof can be found in Appendix~\ref{sec:app-G-central}. 

\begin{definition}
A vertex-disjoint path (one such path need not start from a source or end at a terminal) is a collection of edges. Any two edges in the path can't have the same head or the same tail.
\end{definition}
\begin{restatable}{lemma}{Gdual}
  \label{lem:G-dual}
  Given a Queue-G model $(V, E, \arrivals, \processes)$ with $n$ sources and $m$ servers, the following two conditions are equivalent:
\begin{enumerate}[(1)]

  \item The linear system of the second statement in Theorem \ref{thm:G-central} is feasible.
  \item For any $ \bm{\alpha}\in \{\mathbb{R}_{+}^{n+m} \given \alpha_{i}=0 \ \text{if}\ i \in \terminal\}$, there is a vertex-disjoint path set $U$, such that  $\sum_{(i,j)\in U}(\alpha_{i}-\alpha_{j})\process_{j}>\sum_{i\in \source}\alpha_{i}\arrival_{i}$.
\end{enumerate}
\end{restatable}
% The lemma is an application of Farkas' lemma.  We relegate its proof to 
% \section{Two examples of Queue-G Model}
\subsection{Decentralized Multi-Layer Networks}
\label{sec:G-decentral}

% \subsubsection{Incomplete bipartite Queueing Network}
\subsubsection{System Failure with Myopic Queues}
\label{sec:G-example}

In a queueing system with multiple layers, (i.e., when $\medium \neq \emptyset$), a natural extension of the utility in bipartite systems as defined in Section~\ref{sec:prelim} is to let a queue earn utility 1 at a time step if one of its packet is successfully processed by the server it is sent to.
The hope is that when all queues focus on getting their packets processed by the \emph{next} server, the system runs relatively efficiently.
Unfortunately, as the following example shows, when the queues run no-regret strategies on such utilities, they may be too short-sighted for the decentralized system to have performance comparable to a centralized one, even if one increases the processing capacities by any constant factor.

\begin{figure}[H]
  \centering
  \includegraphics[scale=0.25]{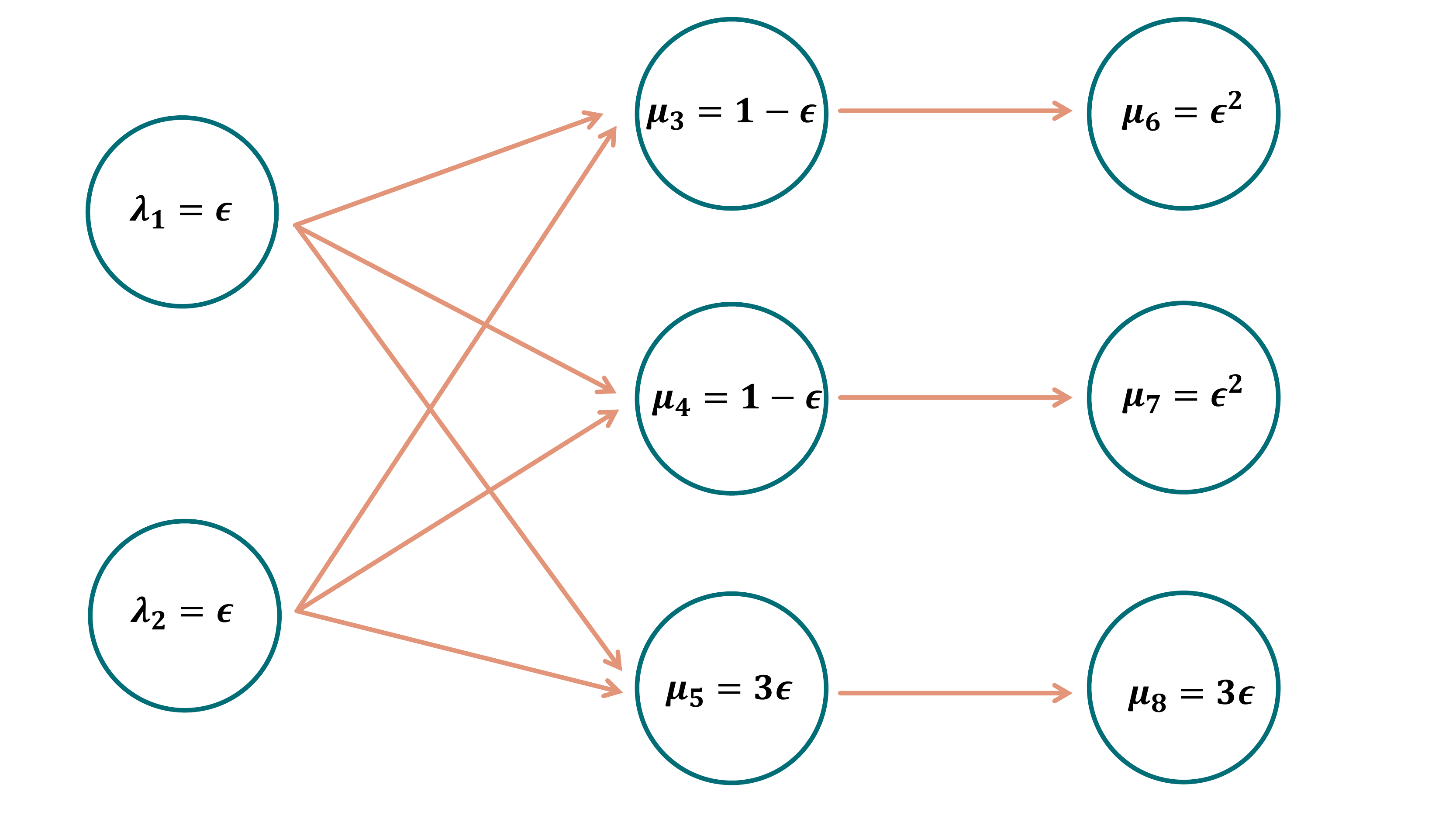}
  \caption{A queueing system with two layers of servers.  A centralized policy sending packets from both sources to server~5 makes the system stable, but the two sources may find it a no-regret strategy to send requests to servers 3 and~4, respectively.}
  \label{fig:G-incomplete}
\end{figure}

\begin{example}
  \label{ex:G-unstable}
  The system shown in Figure~\ref{fig:G-incomplete} is stable under a centralized policy.  
  One feasible solution to the linear system given in Theorem~\ref{thm:G-central} is $z_{15} = z_{25} = 0.4$, $z_{58} = 1$, with all other coordinates of $\mathbf z$ set to 0. 
  It is not difficult to see that, if both queue~1 and queue~2 send their requests to server 3 and server~4, respectively, they play no-regret strategies, but the system is unstable because packets accumulate at servers 3 and~4.
  The phenomenon persists even when the processing capacities are increased by a factor of $\frac 1 \eps$.

\end{example}

\subsubsection{Stability with Queue Length Aware Utilities}
\label{sec:G-decentral-sol}

% Since two motivating examples above showing the unbounded gap between centralized and decentralized systems, we define a new service priority rule in \emph{Queue-G} Model. The new service priority rule at each server is to choose the packet from the longest queue (breaking ties arbitrarily). In previous sections, service priority rule at each server is to choose the oldest packet.

Example~\ref{ex:G-unstable} suggests that the instantaneous, local feedback is not enough to align the queues' interests with the system's efficiency.
In this section we show that, when we incorporate one other piece of local information, the \emph{queue lengths}, into the queues' utilities and the service priority rule, we can recover the bicriteria results we showed for single-layer systems in Section~\ref{sec:IB}.

% In decentralized system of \emph{Queue-G} Model, we also assume that nodes in $\source \cup \medium$ use no-regret learning strategies which is defined in section \ref{stability and no regret}. Here, we define a new utility function:
Recall that $Q_t^i$ denotes the length of queue~$i$ at time~$t$.
Our new utility for queue~$i \in \source \cup \medium$ for sending a request to server~$j$ at time~$t$ is
$$
u_t^i(j, a_{-i}(t)\given \mathcal{F}_{t})=
\left\{
\begin{aligned}
  &Q_t^i-Q_t^j, &&\text{if the packet sent to $j$ is successfully processed;}\\
&0, && \text{otherwise.}
\end{aligned}
\right.
$$

Note that this utility function immediately implies that it is never in a queue's interest to send a request to a server with a queue longer than itself.
  Also recall that $Q_t^j = 0$ for any $j \in \terminal$ at any time~$t$.
% Node $i$ will not send packets to server $j$ if $\queue_t^i<\queue_t^j$. 
% Recall that $\queue_t^i$ is the length of queue $i$ at time step $t$. 
% In comparison, the utility function in \emph{Queue-B} model is the number of packets from node $i$ cleared during time step $t$. 
% We should also specify the content of a history: 
The history $\mathcal{F}_{t}$ now includes information on which packets have been cleared and the queue size $Q_t^i$.\footnote{This is different from the setup in Section~\ref{sec:prelim}.
  We now no longer have the independence between the time interval between packet arrivals and histories prior to the their clearing.  
  As will be clear in the proof, this independence is no longer needed in the proof.
  The introduction of queue lengths makes the change in the potential more directly connected with the queues' utilities.
}

We also change the servers' priority rules to preferring requests from longer queues.  With the new utilities and service priority rules, the sufficient condition we obtain for decentralized stability is:

% In comparison, in \emph{Queue-B} model, $\mathcal{F}_{t}$ only includes information of which packets have been cleared and the age of the currently oldest packet in each node, but does not include the queue size $Q_t^i$.

\begin{assumption}\label{decentralized G stability condition}
  \label{assump:G-decentral}
There is a $\beta>0$ such that for any $ \bm{\alpha}\in \{\mathbb{R}_{+}^{n+m} \given \alpha_{i}=0 \ \text{if}\ i \in \terminal\}$, there is a vertex-disjoint path set $\disjointpathset$ such that 
$$
\frac{1}{2}(1-\beta)\sum_{(i,j)\in \disjointpathset}(\alpha_{i}-\alpha_{j})\process_{j}>\sum_{i\in \source}\alpha_{i}\arrival_{i}
$$
\end{assumption}

\begin{restatable}{theorem}{thmGdecentral}
\label{thm:decentralized G}
\label{thm:G-decentral}
If a Queue-G model queueing system satisfies Assumption \ref{decentralized G stability condition} holds, and nodes use no-regret learning strategies with $\delta=\frac{\beta \process_{(m)}}{96(n+m)^2}$, then the queueing system is strongly stable.
\end{restatable}

A quick comparison between Assumption~\ref{assump:G-decentral} and the dual form of condition for centralized stability in Lemma \ref{lem:G-dual} shows that, a queueing system is guaranteed to be stable with queues using no-regret learning strategies as long as the system can be stable under a centralized policy with twice as many packet arrivals.

% if one has a \emph{Queue-G} Model which is stable under a centralized policy, then, doubling its processing capabilities guarantees its stability when nodes use no-regret strategies. Denote by $\process_{(m)}$ the minimum element of vector $\processes$.

% We introduce a potential function with the intention to apply Theorem~\ref{thm:pemantle} to its square root. Let $\bm{\queue}$ be the vector $(\queue_t^1, \cdots, \queue_t^{n+m})$.
% The potential function is defined as
As in the proof for Theorem~\ref{thm:IB-decentral}, we introduce a similar potential function
\begin{align}
  \Phi\left(\mathbf{\queue}_{t}\right) \coloneqq \frac{1}{2} \sum_{i\in \source \cup \medium} \queue_t^i(\queue_t^i-1),
\end{align}
and define $Z_\ell$ as its square root at the beginning of the $\ell$-th window of length~$w$.
There is change in the proofs for both the negative drift condition and the bounded jump condition.  
We detail here the main different steps in Lemma~\ref{lem:negative-drift-G} for the negative drift condition, and relegate the rest to Appendix~\ref{sec:app-G-decentral}.

% We analyze the system by dividing the time steps into windows of length~$w$ each, for some large enough~$w$.
% Let $Z_\ell \coloneqq \sqrt{\potential\left(\mathbf{Q}_{\ell \cdot w} \right)}$ be the square root of the potential function at the beginning of the $\ell$-th window.
% The main work lies in showing that $(Z_{\ell})_\ell$ satisfies the conditions of Theorem~\ref{thm:pemantle}, which implies $\Ex{Z_\ell^a}$ is bounded for any $a > 0$.
% This in turn implies that $\Ex{(Q_t)^a}$ is bounded.  

\begin{restatable}{lemma}{Gnegativedrift}[Negative drift condition.]
  \label{lem:negative-drift-G}
Let $b=\frac{8\sqrt{2(n+m)}}{\beta  \process_{(m)}}( \sum_{i\in \source}\lambda_i^2 w^2+  \sum_{i\in \medium \cup \terminal} \process_j w^2+w)$, $c=-\frac{\sqrt{2}\process_{(m)}\beta w}{128(n+m)}$. Then $Z_0=0 \leq b$ and, for all $\ell$,
\begin{align*}
    Z_{\ell}>b \quad \Rightarrow  \Ex{Z_{\ell+1}- Z_{\ell} \given Z_0, \cdots, Z_{\ell}} \leq -c.
\end{align*}

\end{restatable}

\begin{restatable}{lemma}{Gboundedjump}[Bounded jump condition.]
  \label{lem:bounded-jump-G}
  For each even integer $p \geq 2$, there is a
  constant $d_p$, such that for all $\ell$,
  \begin{align*}
  \Ex{|Z_{\ell+1}-Z_{\ell}|^p \given Z_0, \cdots, Z_{\ell}} \leq d_p.
  \end{align*}
\end{restatable}

% We present here the key steps of  proof of Lemma \ref{lem:negative-drift-G}, and relegate the rest to Appendix~\ref{sec:app-G-decentral}.

At time step $(\ell+1) \cdot w$, let $\tau_{i}$ be the number of unprocessed packets at node $i$ which arrived before time step $\ell \cdot w$. 
Note the difference from the definition of~$\tau_i$ in the proof of Theorem~\ref{thm:IB-decentral} --- there, $\tau_i$ is the \emph{age} of the oldest packet in queue~$i$ at time~$\ell w$.
For any node $i$ and $t\in [\ell \cdot w, (\ell +1)\cdot w]$,  $\tau_{i}\leq Q_{t}^{i}\leq \tau_{i}+w$. Recall that $C_t^j$ is the indicator variable for server~$j$ succeeding in processing a packet if it picks one up, and $u_t^i$ is the utility of queue $i$ at time step $t$. 
For a server~$j$, define its \emph{contribution} $v_t^j$ to be $u_t^i$ if $j$ successfully clears a packet from queue $i$ at time step $t$, and is 0 if it fails to do so. 
Then, at any time step~$t$, $\sum_{i\in \source \cup \medium} u_t^i=\sum_{j\in \medium \cup \terminal} v_t^j$.
A key observation is that, with the new utility functions, when a packet is cleared, the decrease in potential is exactly equal to the increase in the corresponding queue's utility.
% An interesting observation is when clearing a packet, the utility it brings is equal to the decrease it causes in potential function. If a packet is successfully cleared from node~$i$ by server~$j$, then the utility of node~$i$ is increases by $\queue_t^i-\queue_t^j$. The potential function also decreases by $\queue_t^i-\queue_t^j$ since a packet leaves node~$i$ and joins node~$j$. 
We can therefore calculate the decrease in potential function by tracking the sum of utilities of all nodes.
The following key lemma lower bounds the total utility over a time window as a function of the number of old packets ($\tau_i$'s), assuming the realized processing capacities of all servers are around their expectations. 
Its use of vertex-disjoint paths paves the way for applying Assumption~\ref{assump:G-decentral}.

\begin{lemma}
  \label{lem:lower-bound-utility}
  For any $\eps>0$, if \ $\sum_{t = \ell \cdot w}^{(\ell + 1)\cdot w - 1} C_t^j \geq (1 - \eps) \process_j w$ for each~$j \in \medium \cup \terminal$, then for any set~$\disjointpathset$ of vertex-disjoint paths, $\sum_{i \in \source \cup \medium}\sum_{t = \ell \cdot w}^{(\ell + 1)\cdot w - 1} u_t^i \geq \frac{1}{2}\sum_{(i,j)\in \disjointpathset} (\tau_{i}-\tau_{j}-w)(1-\epsilon) \process_j w- \sum_{i \in \source \cup \medium}\reg_{i}(w,(\ell + 1)\cdot w)$.
\end{lemma}

\begin{proof}

% For any queue~$i$ with $\tau_{i}>0$, it means that queue~$i$ sending packets all the time during this time window. 
  Any queue $i$ with $\tau_i > 0$ is non-empty throughout the time window.
  Let server~$j$ be any server to which $i$ can send requests.
Let $X_{jt}^{\tau_{i}}$ be the indicator variable for the event that \emph{some queue}~$k$ with $Q_t^k\geq \tau_i$ sends a packet to server $j$ at time step $t$. Then at time step $t$ when $X_{jt}^{\tau_{i}}=0$, node $i$'s packet would have been picked up by server $j$ had $i$ sent a request, because no other packet sent to $j$ is from a node whose queue length is at least $\tau_i$ , so the packet from $i$ has priority.  
Over the time window, node $i$'s counterfactual utility, by always sending requests to~$j$, could have been at least $\sum_{t = \ell \cdot w}^{(\ell + 1)\cdot w - 1}C_{t}^{j}(1-X_{jt}^{\tau_{i}})(Q_{t}^{i}-Q_{t}^{j})\geq \sum_{t = \ell \cdot w}^{(\ell + 1)\cdot w - 1}C_{t}^{j}(1-X_{jt}^{\tau_{i}})(\tau_i-Q_{t}^{j})$ since $Q_{t}^{i}\geq \tau_i$.
By definition of regret, we have
\begin{align*}
\sum_{t = \ell \cdot w}^{(\ell + 1)\cdot w - 1} u_t^i 
&\geq \sum_{t = \ell \cdot w}^{(\ell + 1)\cdot w - 1} C_{t}^{j}(1-X_{jt}^{\tau_{i}})(\tau_i-Q_{t}^{j})-\reg_{i}(w,(\ell + 1)\cdot w).
\end{align*}
On the other hand, whenever $X_{jt}^{\tau_{i}}C_{t}^{j}=1$, server $j$ clears a packet from a queue with at least $\tau_{i}$ packets at time step $t$. Then $v_t^j$ is at least $\tau_i-Q_{t}^{j} \geq \tau_{i}-\tau_{j}-w$ since $Q_{t}^{j}\leq \tau_j+w$.
Over the time window, $\sum_{t = \ell \cdot w}^{(\ell + 1)\cdot w - 1} v_t^j \geq \sum_{t = \ell \cdot w}^{(\ell + 1)\cdot w - 1} C_{t}^{j}X_{jt}^{\tau_{i}}(\tau_i-Q_{t}^{j})$.

For each pair of queue~$i$ with $\tau_i>0$ and server~$j$ that can serve $i$, we have
\begin{align}\label{eq: lower bound utility of a pair}
\sum_{t = \ell \cdot w}^{(\ell + 1)\cdot w - 1} u_t^i+\sum_{t = \ell \cdot w}^{(\ell + 1)\cdot w - 1} v_t^j
&\geq \sum_{t = \ell \cdot w}^{(\ell + 1)\cdot w - 1}(\tau_{i}-\tau_{j}-w)C_t^j-\reg_{i}(w,(\ell + 1)\cdot w).
\end{align}
Now we are ready to use Assumption~\ref{decentralized G stability condition}. Let $\alpha_i= \tau_i$ for each $i$. For each node $i \in \terminal$, $\alpha_i=\tau_i=0$ since terminal nodes have no queues. By Assumption \ref{decentralized G stability condition}, we can find a vertex-disjoint path $\disjointpathset$ such that 
\begin{align}\label{eq:decentralized-G condition}
     \frac{1}{2}(1-\beta)\sum_{(i,j)\in \disjointpathset}(\tau_{i}-\tau_{j})\process_{j}>\sum_{i\in \source}\tau_{i}\arrival_{i}.
 \end{align}
Now, we are ready to give a lower bound on the total utilities:
\begin{align*}
2\sum_{i\in \source \cup \medium} \sum_{t = \ell \cdot w}^{(\ell + 1)\cdot w - 1} u_t^i 
=&\sum_{j \in \source \cup \medium} \sum_{t = \ell \cdot w}^{(\ell + 1)\cdot w - 1} u_t^i+\sum_{i\in \medium \cup \terminal} \sum_{t = \ell \cdot w}^{(\ell + 1)\cdot w - 1} v_t^j\\
\geq& \sum_{(i,j)\in \disjointpathset} \left(\sum_{t = \ell \cdot w}^{(\ell + 1)\cdot w - 1} u_t^i+\sum_{t = \ell \cdot w}^{(\ell + 1)\cdot w - 1} v_t^j \right)\\
\geq&  \sum_{(i,j)\in \disjointpathset} \sum_{t = \ell \cdot w}^{(\ell + 1)\cdot w - 1}(\tau_{i}-\tau_{j}-w)C_t^j-\sum_{i \in \source \cup \medium}\reg_{i}(w,(\ell + 1)\cdot w)\\
\geq& \sum_{(i,j)\in \disjointpathset} (\tau_{i}-\tau_{j}-w)(1-\epsilon) \process_j w- \sum_{i \in \source \cup \medium}\reg_{i}(w,(\ell + 1)\cdot w),
\end{align*}
% \geq& \frac{\beta}{2} w \sum_{(i,j)\in \disjointpathset}(\tau_{i}-\tau_{j})\process_j+2(1+\epsilon)w \sum_{i\in \source}\arrival_{i}\tau_{i}\\
% &-w^2 \sum_{j \in \medium \cup \terminal}\process_j- \sum_{i \in \source \cup \medium}\reg_{i}(w,(\ell + 1)\cdot w),
where the second inequality uses \eqref{eq: lower bound utility of a pair}, and the last inequality is due to $\sum_{t = \ell \cdot w}^{(\ell + 1)\cdot w - 1}C_t^j \geq (1-\epsilon)\process_j w$ which is the condition of the lemma.
% and the last inequality uses \eqref{eq:decentralized-G condition} and $\epsilon=\frac{\beta}{8}, 1-\frac{\beta}{8}-(1+\frac{\beta}{8})(1-\frac{\beta}{8})\geq \frac{\beta}{2}$.
\end{proof}

We sketch the rest of the proof, and relegate details to Appendix~\ref{sec:app-G-decentral}.
When the servers' realized processing capacities are close to their expectations (as in the condition of Lemma~\ref{lem:lower-bound-utility}) and when the queues' regrets are small (which should happen with high probability by assumption), the lower bound given by Lemma~\ref{lem:lower-bound-utility} on the utility of queues implies a lower bound on the decrease in the potential function due to packet processing (Lemma~\ref{lem:phi-decrease-G}). 
% Using Assumption~\ref{assump:G-decentral}, we change the formulation of potential decrease from $\sum_{(i,j)\in \disjointpathset}(\tau_i-\tau_j) \process_j$ to $\sum_{i\in \source} \arrival_i \tau_{i}$, in order to eliminate the term of potential increase which will introduced in the following.
When the packet arrivals in the sources are close to their expectations, we can further bound the increase in the potential function due to packet arrival (Lemma~\ref{lem:phi-increase-G}).
We define an event~$B$ (Definition~\ref{def:eventB}), which happens with high probability (Claim~\ref{claim:eventB}), and under which all of these events (of concentration and no regret) happen.
Roughly speaking, when $B$ happens, the increase in the potential due to packet arrivals is at least offset by the decrease due to packet clearing when we use Assumption~\ref{assump:G-decentral} to generate an appropriate~$U$ in Lemma~\ref{lem:lower-bound-utility} and relate the term $\frac 1 2 \sum_{(i, j) \in U} (\tau_i - \tau_j) \process_j w$ to $\sum_{i \in \source} \tau_i \arrival_i$.
With a small probability, event~$B$ does not happen, and it is relatively straightforward to upper bound the increase in the potential in this case.
% (Most pessimistically, the queue length of each source node grows by~$w$.)

To put things formally, recall that $\tau_i$ is the length of packets in queue $i$ at time step $(\ell+1)\cdot w$, which arrived in queue $i$ before time step $\ell \cdot w$; let $\bm{\tau}=\{\tau_1, \cdots, \tau_n\}$. Given $\tau_i$ for each node $i$, by Assumption~\ref{assump:G-decentral}, let $\disjointpathset^*$ be the vertex-disjoint path such that 
\begin{align*}
\frac{1}{2}(1-\beta)\sum_{(i,j)\in \disjointpathset^*}(\tau_{i}-\tau_{j})\process_{j}>\sum_{i\in \source}\tau_{i}\arrival_i.
\end{align*}

\begin{restatable}{lemma}{Gphidecrease}
    \label{lem:phi-decrease-G}
  Under \emph{event B}, $\Phi(\mathbf{Q_{\ell \cdot w}})-\Phi(\bm{\tau}) \geq \frac{\beta}{4}\sum_{(i,j)\in \disjointpathset^*}(\tau_i-\tau_j)
  \process_j w+(1+\frac{\beta}{8})  \sum_{i\in \source} \arrival_i \tau_{i}w -\sum_{i\in \medium \cup \terminal} \process_i w^2-w$.
\end{restatable}

\begin{restatable}{lemma}{Gphiincrease}
  \label{lem:phi-increase-G}
  Under \emph{event B}, $\Phi(\mathbf{Q_{(\ell+1) \cdot w}})-\Phi(\bm{\tau}) \leq (1+\frac{\beta}{8}) \sum_{i\in \source} \lambda_{i}\tau_{i} w+ \sum_{i\in \source}\lambda_{i}^{2} w^{2}$.
\end{restatable}

With a small probability, event~$B$ does not happen, and it is relatively straightforward to upper bound the increase in the potential in this case. (Most pessimistically, no packets is cleared during the time window and for each time step, there is a packet arriving at each source node, then the queue length of each source node grows by~$w$.)

\begin{restatable}{lemma}{notB}
  \label{lem:not-B}
  If event~$B$ does not happen, then $\Phi(\mathbf{Q_{(\ell+1) \cdot w}})-\Phi(\mathbf{Q_{\ell \cdot w}}) \leq \sum_{i\in \source} \queue_{\ell \cdot w}^i w +w^2$.
\end{restatable}

Lemma~\ref{lem:negative-drift-G} follows from combining Lemma \ref{lem:phi-decrease-G}, \ref{lem:phi-increase-G} and~\ref{lem:not-B}.

\section{Patient Queueing Model}
\label{sec:patient}

In this section we extend a model introduced in \citep{gaitonde2021virtues}, where queues do not vary their routing policies from step to step, but evaluate the utility of a fixed routing policy over a long period of time.
% The game for the queues can be seen as a single round one, even though its payoffs are calculated as the expected utility from many time steps.   Nash equilibria can be defined in the resulting game.
On complete bipartite graphs, \citeauthor{gaitonde2021virtues} showed that Nash equilibria always exist in the resulting game, and that a system is stable under \emph{any} Nash as long as it can be stable under a central policy with $\frac e {e - 1}$ times as much arrival rates.
% given any system that is stable under some centralized policy, one may enlarge its processing capacities by a factor of no more than $\frac e {e - 1}$ to guarantee that it is stable under any Nash equilibrium defined in this sense, and this factor is tight.
We obtain a similar result for incomplete bipartite graphs, but the factor in our bicriteria result worsens to~2.
% but we are only able to show for such general cases a worse factor of 2.  
We leave for future work to decide whether this factor can be improved.

Below we first describe the model in more detail, before presenting our result.

% Patient Queueing Model is a decentralized system of  \emph{Queue-B} Model. Each queue's strategy is a fixed independent distribution over servers. In comparision, queues use no-regret learning strategy in previous sections. By defining a cost function which is the long-run aging rate for each queue, there exists Nash Equilibrium in Patient Queueing Model. In this section, we find the stability condition for any Nash Equilibrium and compare it to the centralized stability condition.

\subsection{Model Description}
\label{sec:patient-model}

% A bipartite \emph{Patient Queueing Model} is given by a tuple $(\source \cup \terminal, E, \arrivals, \processes)$ where $(\source \cup \terminal, E)$ constitutes a bipartite graph, and 
A bipartite \emph{Patient Queueing Model} has the same packet arrival, routing and processing procedures as in a Queue-B model described in Section~\ref{sec:prelim}; the servers' priority rule is to pick the oldest packet.  
The main difference is that the queues are ``patient'': each queue fixes a routing strategy in the form of a distribution over the servers it can reach, and evaluates its utility/cost over a long time period.  
Formally, the strategy space of queue~$i$ is 
 the simplex over the servers $i$ can send requests to: $\Delta_i \coloneqq \Delta(\outneighbor(i))$.
% $\Delta_{i} \coloneqq \{p_i \in \mathbb{R}_+^{n} \given \sum_j p_{ij}=1, p_{ij}=0 \ \text{if} \ (i,j)\notin E \}$.
By adopting a strategy $p_i \in \Delta_i$, queue~$i$ in each round samples a server according to the distribution given by~$p_i$, independently of all history and other happenings in the system, and sends a request to the sampled server if its queue is non-empty.
Let $\stratsmi$ be the strategies of the queues other than~$i$, then the \emph{cost} of queue~$i$ for using strategy~$\strati$ is 
$c_i(p_i, \stratsmi) \coloneqq \lim_{t \to \infty} \frac{T_t^i}{t}$, where $\age_t^i$ is the age of the oldest packet in queue~$i$ at time step~$t$.
Each queue aims to minimize its cost.  
A strategy profile $\mathbf{p}$ is a \emph{Nash equilibrium} if for each queue~$i$, $p_i \in \argmin_{p{'} \in \Delta_i} c_i(p{'}, \stratsmi)$.

% Denote by $G=(\source\cup \terminal, E)$ the bipartite graph induced by queueing system. Let $p_i$ denote the strategy of queue $i$ and $\Delta_{i}$ denote the strategy space of queue $i$. Then, $ \Delta_{i}=\{p_i\in \mathbb{R}^{n} \given \sum_j p_{ij}=1, p_{ij}=0 \ \text{if} \ (i,j)\notin E \}$, where $p_{ij}$ represents the probability that queue $i$ chooses server $j$. Denote by $\mathbf{p}$ the strategy profile of all queues. Then, $\mathbf{p}=\{p_1,\cdots, p_n\} \in \Delta_1 \times\cdots \times \Delta_n$. Then, we define the cost function of queue $i$:
% \begin{align*}
%      c_i(p_i, p_{-i}) \coloneqq \lim_{t \to \infty} \frac{T_t^i}{t},
% \end{align*}
% where $T_t^i$ represents the age of oldest packet in queue $i$ at time step $t$. Each queue $i$ wants to minimize its cost function $c_i$. The strategy profile $\mathbf{p}$ is a \emph{Nash Equilibrium} if for each queue $i$,
% \begin{align*}
%     p_i = \arg \min_{p^{'} \in \Delta_i} c_i(p^{'}, p_{-i})
% \end{align*}

\citet{gaitonde2021virtues} considered a special case of the Patient Queueing Model, where the underlying bipartite graph is complete, i.e., $E = \source \times \terminal$. 
 % \citeauthor{gaitonde2021virtues} 
 They gave an algorithm that, given a strategy profile~$\mathbf{p}$, computes an $r_{i}(\mathbf{p})$ for each queue~$i$, with $r_{i}(\mathbf{p})$ equal to $c_{i}(\mathbf{p})$ almost surely. 
This algorithm played a crucial role in the derivation of their main result.
The algorithm extends directly to our more general setting, and is also a key step in our result in this section.
We present this algorithm next.

% In the following, we will introduce the simple algorithm and some useful results in \cite{gaitonde2021virtues} which can be directly extended to our model.

\subsection{Gaitonde and Tardos's Algorithm for Computing Costs}
\label{sec:patient-alg}

Algorithm~\ref{alg:compute-cost} is 
a straightforward generalization (to general bipartite graphs) of \citeauthor{gaitonde2021virtues}'s algorithm for computing the queues' costs.  
We give some rough intuition here.
Thanks to the service priority rule, 
% a queue that accumulates older packets has priority.
in the long run, the queues are tiered according to the rates of growth of their lengths: the faster growing ones have higher priority over the slower growing ones.
Determining which queues grow the fastest is like a self-fulfilling prophecy: a group of queues grow the fastest even when they have the highest priority.
The algorithm enumerates all possible ``first tier'' queues, and finds the one that fulfills the ``prophecy.'' 
It then continues to find lower tiers, assuming that all higher tiers have priority.
Nailing down this intuition involves intricate probabilistic arguments, and is a major technical accomplishment of \citep{gaitonde2021virtues}.

% In this algorithm, the $n$ queues are partitioned into $K$ groups.  % for queues in the same group, they have the same output $r_i(\mathbf{p})$. 
% Arrival rates $\arrivals = (\arrival_1, \cdots, \arrival_n)$ and processing rates $\processes = (\process_1, \ldots, \process_m)$. 
\RestyleAlgo{ruled}
\begin{algorithm}
\caption{Computing the queues' costs given their strategies}
\label{alg:compute-cost}
Input: Queueing system $(\source \cup \terminal, E, \arrivals, \processes)$,  strategy profile $\mathbf{p}$\\
$I \gets \source, k \gets 1$\\
\While {$I \neq \emptyset$}{
  Compute for each $Q \subseteq I$, % the minimum ratio of $f(Q)$ for each $Q \subseteq I$\\
% \begin{align*}
  $ f(Q)=\frac{\sum_{j=1}^{m}\process_j(1-\prod_{i\in Q}(1-p_{ij}))}{\sum_{i\in Q} \arrival_i}$. \\
% \end{align*}
  Let $Q_k$ be the minimizer of $f(Q)$, breaking ties in favor of larger cardinality.
  % \footnote{\citet{gaitonde2021virtues} showed that, if $U_1$ and $U_2$ are both minimizers of~$f$, then so is $U_1 \cup U_2$, so there is always a unique minimizer of the largest cardinality.} \\
\eIf{$f(Q_{k})\geq 1$}{For any $i\in I$, output $r_i(\mathbf{p})=0$, terminate.}
{For each $i\in Q_{k}$,  $r_i(\mathbf{p})=1-f(Q_{k})$;\\
Update $\process_j \gets \process_j \prod_{i\in Q_{k}}(1-p_{ij}), I\gets I\setminus Q_k, k \gets k+1$.}
}
Output: a sequence of queueing groups $Q_1,\cdots,Q_K$ and $r_i(\mathbf{p})$ for each queue $i$.
\end{algorithm}

In the step that picks $Q_k$, there is no ambiguity because the union of minimizers of $f(Q)$ can be shown to be another minimizer (see Lemma~\ref{closure property}).
The following results are straightforward generalizations of corresponding results in \citep{gaitonde2021virtues}.  We state them without proofs.
Theorem \ref{thm:cost} and~\ref{thm:nash-exists} correspond to Theorem 4.1 and 3.3 in \citep{gaitonde2021virtues}, respectively.
Theorem~\ref{thm:q1} is a generalization of Lemma 3.3 and Theorem 4.4 in \citep{gaitonde2021virtues}.

% Now, we present some useful results in \cite{gaitonde2021virtues} which can be extended to our model directly. The first one is cost function $c_i(\mathbf{p})$ is equal to the output of algorithm $r_i(\mathbf{p})$ almost surely.
\begin{theorem}% (Theorem 4.1 in \cite{gaitonde2021virtues})
  \label{thm:cost}
For any strategy profile $\mathbf{p}$, $c_i(\mathbf{p})=\lim_{t \to \infty} \frac{T_t^i}{t}=r_i(\mathbf{p})$ almost surely.
\end{theorem}
% The second one is the existence of Nash Equilibrium in this \emph{Patient Queueing Model}.

\begin{theorem}% (Theorem 3.3 in \cite{gaitonde2021virtues})
  \label{thm:nash-exists}
  If the cost function of each queue~$i$ is defined to be $r_i(\strats)$, % With cost function $r_i(\mathbf{p})$ for each queue $i$,
  then every system of the Patient Queueing Model admits a Nash Equilibrium.
\end{theorem}

% Combining Lemma 3.3 and Theorem 4.4 in \citet{gaitonde2021virtues}, we have the following theorem:
\begin{theorem}\label{thm:oldest group}
  \label{thm:q1}
  Given a queueing system and a strategy profile~$\mathbf{p}$, let $Q_1$ be the first group of queues output by Algorithm~\ref{alg:compute-cost}.
  If $f(Q_1)>1$, then the queueing system is stable under~$\strats$.
\end{theorem}
% In the next subsection, we will prove that if we increase processing capabilities by a small factor, then, for any Nash Equilibrium strategy profile $\mathbf{p}$, the queueing system is stable. 
% \begin{lemma}(Monotonicity)
% Given strategy profile $\boldsymbol{p}$, for any queue $i$ in group $S_k$, queue $i^{'}$ in group $S_{k^{'}}$, if $k\leq k^{'}$, then $r_i(\boldsymbol{p})\geq r_{i^{'}}(\boldsymbol{p})$
% \end{lemma}

\subsection{Price of Anarchy in Patient Queueing Model}
\label{sec:patient-poa}

Our bicriteria result for general bipartite graphs is presented again in a form more comparable to the dual form of conditions for centralized stability (see Lemma~\ref{lem:IB-dual}).

\begin{theorem}
  \label{thm:patient}
  Given a bipartite queueing system $(V, E, \arrivals, \processes)$, if for any $\boldsymbol{\alpha}=(\alpha_1,\cdots,\alpha_n)\in \{0,1\}^n$, there is a matching matrix $M$, such that $\frac{1}{2} \boldsymbol{\alpha}^\top M \processes > \boldsymbol{\alpha}^\top \arrivals$, then the system is stable in any Nash Equilibrium in the patient queue model.
\end{theorem}

% Recall from Lemma~\ref{lem:IB-dual} that a queueing system on a bipartite graph is stable under a centralized policy if and only if for any $\alphas \in \mathbb R_+^n$, there is a matching matrix $M$ such that $\alphas^\top M \processes > \alphas^\top \arrivals$.
% It can be seen that, with the processing capacities doubled, such a system is guaranteed to meet the sufficient condition given in Theorem~\ref{thm:patient}.

We remark that repeatedly playing a Nash Equilibrium strategy profile $\mathbf{p}$ may not be no-regret strategies (see Example~\ref{ex:nash-regret}), therefore Theorem~\ref{thm:patient} is not implied by our results in Section~\ref{sec:IB}.

\begin{proof}
Let $\strats$ be any Nash equilibrium in a system satisfying the condition in the theorem.
By Theorem \ref{thm:oldest group}, we only need to consider the first group of queues $Q_1$ output by Algorithm~\ref{alg:compute-cost} and prove $f(Q_1)> 1$. 
Let $\bm{\alpha}$ be the indicator vector for the set $Q_1$, i.e., $\alpha_i=1$ if $i\in Q_1$, and $\alpha_i=0$ otherwise. 
By assumption, there is a matching matrix $\matchingmatrix$ with $\frac{1}{2} \boldsymbol{\alpha}^\top M \processes > \boldsymbol{\alpha}^\top \arrivals$. 
For each $i\in Q_1$, let $s_i$ matched to~$i$ by~$M$, i.e., 
% be the server that 
$M(i,s_{i})=1$;
if no server is matched to~$i$, 
% If for any server $j \in [m],M(i,s_{i})=0$, we 
let $s_i=0$ and set $\process_{0}=0$. 
Then we have 
$ \frac{1}{2}\sum_{i\in Q_1}\process_{ s_{i}} >\sum_{i\in Q_1}\arrival_i$.
By definition of $f(Q_1)$,
\begin{align}
    f(Q_1)
    &=\frac{\sum_{j=1}^{m}\process_j(1-\prod_{i\in Q_1}(1-p_{ij}))}{\sum_{i\in Q_1} \arrival_i} 
    % \nonumber \\
    \geq \frac{\sum_{i\in Q_1}\process_{ s_{i}}(1-\prod_{i\in Q_1}(1-p_{is_{i}}))}{\sum_{i\in Q_1} \arrival_i} \nonumber \\
     &=\frac{\sum_{i\in Q_1}\process_{s_{i}}}{\sum_{i\in Q_1} \arrival_i} - \frac{\sum_{i\in Q_1}\process_{s_{i}}\prod_{i\in Q_1}(1-p_{is_{i}})}{\sum_{i\in Q_1} \arrival_i} \label{eq:primal ratio},
\end{align}
where in the inequality, we drop the contribution of processing powers from servers not matched to any queue in~$Q_1$.
% some non negative terms on the numerator related to servers which are not matched to any queue in $Q_1$.
We claim 
\begin{align} \label{eq:marginal ratio}
    f(Q_1)\geq \frac{\sum_{i\in Q_1}\process_{s_i} \prod_{i\in Q_1}(1-p_{is_i})}{\sum_{i\in Q_1} \arrival_i}.
\end{align}
Combining \eqref{eq:primal ratio} and \eqref{eq:marginal ratio}, we immediately have 
\begin{align}
    f(Q_1) \geq \frac{1}{2} \cdot \frac{\sum_{i\in Q_1}\process_{s_i}}{\sum_{i\in Q_1} \arrival_i} >1.
\end{align}
By Theorem \ref{thm:oldest group}, the queueing system is stable since $f(Q_1)>1$.

It remains to prove \eqref{eq:marginal ratio}.
We say a subset $Q$ of queues is \emph{tight} if $f(Q)=f(Q_1)$. 
By Algorithm~\ref{alg:compute-cost}, tight subsets are minimizers of the function~$f$.
% is the subset that has the smallest $f(Q)$ over all subsets $Q$ of $S_1$.
In the following, we use Nash property. For any queue~$i \in Q_1$, given $p_{-i}$, $p_{i}$ is the strategy that locally maximizes the minimization of  all subset $Q \ni i$, where tight subsets including $i$ are minimizers. If queue $i$ has some deviation from server $j$ to server $j'$ causing $f$ value of all tight subsets including $i$ increase, then it means that this new strategy has a lower aging rate compared to Nash strategy profile, which violates Nash property. 
Since $\mathbf{p}$ is a Nash Equilibrium, for each queue~$i \in Q_1$, for any two servers $j$ with $p_{ij}>0$ and~$j{'}$ that can serve queue~$i$, there is a tight subset~$Q$ such that decreasing $p_{ij}$ and increasing  $p_{ij'}$ by the same (small enough) amount does not increase $f(Q)$, 
\hufu{this is unclear.  why?  Do you mean $Q \ni i$?  Are we considering $i \in Q_1$?  do we need to require $p_{ij} > 0$?  The reader has to guess what other information is missing..}
which means $\frac{\partial f(Q)}{\partial p_{ij}} \geq\frac{\partial f(Q)}{\partial p_{ij'}}$. 
Writing these partial derivatives and dropping the same denominators on both sides, we have
\begin{align*}
    \process_j\prod_{r\in Q\setminus i}(1-p_{rj})\geq
    \process_{j{'}}\prod_{r\in Q \setminus i}(1-p_{rj{'}}).
\end{align*}
Apply this argument to for each $i \in Q_1$ with server~$s_i$ taking the place of $j'$, we have that, for every server $j$ with $p_{ij}>0$, there is a tight subset $Q_i^j$ such that
\begin{align}\label{eq:Nash property}
    \process_j\prod_{r\in Q_i^j\setminus i}(1-p_{rj})\geq
    \process_{s_{i}}\prod_{r\in Q_i^j\setminus i}(1-p_{rs_{i}}).
\end{align}

\citet{gaitonde2021virtues} showed that nonempty intersections of tight subsets are also tight:
\begin{lemma}(Lemma 3.2 in \cite{gaitonde2021virtues})\label{closure property}
If $Q_i$ and $Q_j$ are tight subsets and $Q_i \cap Q_j \neq \emptyset$, then $Q_i \cap Q_j$ and $Q_i \cup Q_j$ are also tight.
\end{lemma}

Define $Q_i=\cap_{j:p_{ij}>0} Q_i^j$. 
$Q_i$ is non empty since $i \in Q_i^j$ for each $j$. Therefore, $Q_i$ is also tight. 
By definition of $f(Q)$, the numerator of $f(Q_i)$ is $\sum_{j=1}^{m}\process_j(1-\prod_{r\in Q_i}(1-p_{ij}))$ and the numerator of $f(Q_i \setminus i)$ is $\sum_{j=1}^{m}\process_j(1-\prod_{r\in Q_i \setminus i}(1-p_{ij}))$. The difference between these two numerators is :
\begin{align*}
  &\sum_{j=1}^{m}\process_j(1-\prod_{r\in Q_i}(1-p_{ij}))-\sum_{j=1}^{m}\process_j(1-\prod_{r\in Q_i \setminus i}(1-p_{ij}))\\
=&\sum_{j=1}^{m} \process_j\left(\prod_{r\in Q_i \setminus i}(1-p_{ij})-\prod_{r\in Q_i}(1-p_{ij})\right)\\
=&\sum_{j=1}^{m} \process_j p_{ij}\prod_{r\in Q_i \setminus i}(1-p_{ij})
\end{align*}
The difference between denominators of $f(Q_i)$ and $f(Q_i \setminus i)$ is $\arrival_i$. Then, we have 
 \hufu{why?}
\begin{align*}
    f(Q_i)=f(Q_i\setminus i)\oplus \frac{\sum_{j=1}^{m} p_{ij}\process_j\prod_{r\in Q_i\setminus i}(1-p_{ij})}{\arrival_i},
\end{align*}
where operation $\oplus$ is defined as $\frac{a}{b}\oplus \frac{c}{d}\coloneqq \frac{a+c}{b+d}$ when $b+d \neq 0$.
Since $Q_i$ is tight, $f(Q_i) \leq f(Q_i\setminus i)$, which means, 
\begin{align*}
    f(Q_i)
    &\geq \frac{\sum_{j=1}^{m} p_{ij}\process_j
    \prod_{r\in Q_i\setminus i}(1-p_{rj})}{\arrival_i}\\
    &\geq \frac{\sum_{j=1}^{m}p_{ij}\process_j
    \prod_{r\in Q_i^j\setminus i}(1-p_{ij})}{\arrival_i}\\
     &\geq \frac{\sum_{j=1}^{m} p_{ij} \process_{s_{i}}\prod_{r\in Q_i^{j}\setminus i}(1-p_{i s_{i}})}{\arrival_i}\\
     &\geq \frac{\process_{s_i} \prod_{i \in Q_1}(1-p_{is_i})}{\arrival_i},
\end{align*}
 where the first inequality holds because if $A=B \oplus C$ and $A\leq B$, then $A\geq C$,  the second inequality is true because $Q_i \subseteq Q_{i}^{j}$, the third inequality follows from~$\eqref{eq:Nash property}$, and the last inequality holds because $Q_{i}^{j} \subseteq Q_1$ and $\sum_{j=1}^{m} p_{ij}=1$.
 Let $a_i=\process_{s_i} \prod_{i \in Q_1}(1-p_{is_{i}})$. 
 We know $f(Q_1)=f(Q_i)\geq \frac{a_i}{\arrival_i}$ for each $i$. Then,
 \begin{align*}
f(Q_1) &\geq \max_{i \in Q_1} \frac{a_i}{\arrival_i} 
\geq \frac{\sum_{i \in Q_1} a_i}{\sum_{i \in Q_1} \arrival_i }\\
&=\frac{\sum_{i\in Q_1}\process_{s_i} \prod_{i\in Q_1}(1-p_{is_i})}{\sum_{i\in Q_1} \arrival_i}.
 \end{align*}
This concludes the proof of \eqref{eq:marginal ratio}.
\end{proof}

\section{Conclusion}
\label{sec:conclusion}
In this work generalize the decentralized queueing systems proposed by Gaitonde and Tardos \cite{gaitonde2020stability,gaitonde2021virtues}.
% In the appendix, 
In the appendix,
we also consider two more variants of the model with multiple layers, when packet arrivals are adversarial instead of probabilistic, and when packets themselves (rather than the queues they are from) determine which servers can process them.
We show that the bicriteria results under no regret strategies are robust against these model modifications.
Lastly, we provide a slightly tighter analysis for the Queue-CB model of~\cite{gaitonde2020stability}
in the full version of the paper.
% in Appendix~\ref{sec:app-CB}. 

Since the bicriteria result in~\cite{gaitonde2020stability} for queues using no regret strategies is tight even for complete bipartite graphs, our results for the more general cases are also tight.  
On the other hand, we do not know if the factor 2 is tight in our result for patient queues in general bipartite graphs.  
It seems challenging to directly apply the deformation technique developed in~\cite{gaitonde2021virtues} in this more general setting; we leave for future work to investigate the tight bound of this problem.

% For patient queueing model, \citet{gaitonde2021virtues} showed that the bicriteria bound is $\frac{e}{e-1}$. We generalize this setting and  and consider scenarios where not all servers can serve all queues, showing that bicriteria bound is no more than 2. It is an interesting open question that whether our results for patient queueing model can be improved and whether bicriteria bounds hold when we generalize patient queueing model and consider the cases when packets need to go through more than one layer of servers before their completions.

\bibstyle{plain}
\bibliography{ref}

\appendix

\section{Missing Proofs from Section~\ref{sec:IB-central}}
\label{sec:app-IB-central}
\IBcentral*
\begin{proof}
We first show sufficiency of the condition. 
By Birkhoff-von Neumann Theorem, any fractional matching matrix~$P$ can be decomposed as a convex combination of matching matrices.
By drawing from this distribution a matching at each time step and letting each queue send its packet to the server it is matched to, there is never contention at any server, and each queue~$i$'s packet is cleared with probability $(P\processes)_i > \lambda_i$.
The random variables $\{Q_t^i\}_t$ clearly satisfies both the bounded jump condition and the negative drift condition in Theorem~\ref{thm:pemantle}, and stability follows.
 
Next, we show the condition is also necessary for the system's stability.  Suppose the condition does not hold, we show that the system cannot be stable under any central policy.
Fix a $T$ large enough.
Let $X_T^i$ denote the total number of arrivals seen by queue~$i$ at the end of time step~$T$, then with constant probability (where the constant depends on $\lambda_i$ but not on~$T$), $X_T^i > \arrival_i T + \sqrt T$, which follows from either the central limit theorem or properties of binomial distributions.
By the independence of arrivals at all queues, with constant probability~$c_0$ (where $c_0$ depends on $\arrival_1, \cdots, \arrival_n$ but not on~$T$), for every~$i$, we have $X_T^i > \arrival_i T + \sqrt T$.
  Denote this event by $A_T$.
  As we have argued, it is without loss of generality to assume that a central policy samples a matching at time step~$t$, where the distribution may depend on $\history_t$.
  Let $M_t$ be the corresponding matching matrix, and let $P$ be $\frac 1 T \sum_{t = 1}^T M_t$, then $M$ is a fractional matching matrix with probability 1.
  Let $P$ be $\Ex{M \given A_T}$.  Then $P$ is a fractional matching matrix.
  By assumption, there exists a queue~$i^*$ for which $(P \processes)_{i^*} \leq \arrival_{i^*}$.
  Then the expected length of queue~$i$ at the end of time~$T$ is at least 
  \begin{align*}
    \Prx{A_T} \cdot \Ex{Q_T^{i^*} \given A_T} & \geq c_0 \left( \Ex{X_T^{i^*} \given A_T} - \Ex{\sum_{t = 1}^T (M_t \processes)_{i^*} \given A_T} \right) \\
    & > c_0 (\arrival_{i^*} T + \sqrt T - (P \processes)_{i^*} T) > c_0\sqrt T.
  \end{align*}
  In the second inequality, we make use of the fact that, at any time step, whether a server succeeds in processing a packet if it picks one up is independent from the routing decisions. 

  This calculation is for every $T$ large enough, so $\Ex{Q_T} > c_0 \sqrt T$ for each large enough~$T$.  
  This shows that the system is not stable.
\end{proof}

\IBdual*

\begin{proof}
For any fractional matching matrix $\probabilitymatrix$, there is a probability vector $\boldsymbol{x}$ over matching matrices, such that $\probabilitymatrix=\sum_{i} x_{i}\matchingmatrix_{i}$. Then, condition (1) can be written as:
\begin{align*}
    \sum_{i}x_{i}\matchingmatrix_{i} \boldsymbol{\process} &\succ \boldsymbol{\arrival}\\
    \sum_{i}x_{i} &\leq 1
\end{align*}
Let $\alpha_i$ denote the dual variable for the first inequality and $\beta$ denote the dual variable for the second inequality. According to Farkas' Lemma, the above linear system is feasible if and only if the following linear system is infeasible:
\begin{equation}\label{dual infeasible}
\begin{aligned}
    \beta-\boldsymbol{\alpha}^\top M \processes&\geq 0, \quad \forall \matchingmatrix \in \matchingmatrixset \\
    \beta-\boldsymbol{\alpha}^\top \arrivals&<0
\end{aligned}
\end{equation}
We can see linear system \eqref{dual infeasible} is infeasible if and only if there is a matching matrix $M$, such that $\boldsymbol{\alpha}^\top M \processes > \boldsymbol{\alpha}^\top \arrivals$, which is condition (2). Both condition (1) and condition (2) hold if and only if linear system \eqref{dual infeasible} is infeasible, then condition (1) and (2) are equivalent.
\end{proof}

\section{An example from Section~\ref{sec:IB-central}}
\label{ex:fractional-dual}
\label{sec:app-example}
\begin{figure}[H]
    \centering
    \includegraphics[scale=0.27]{
    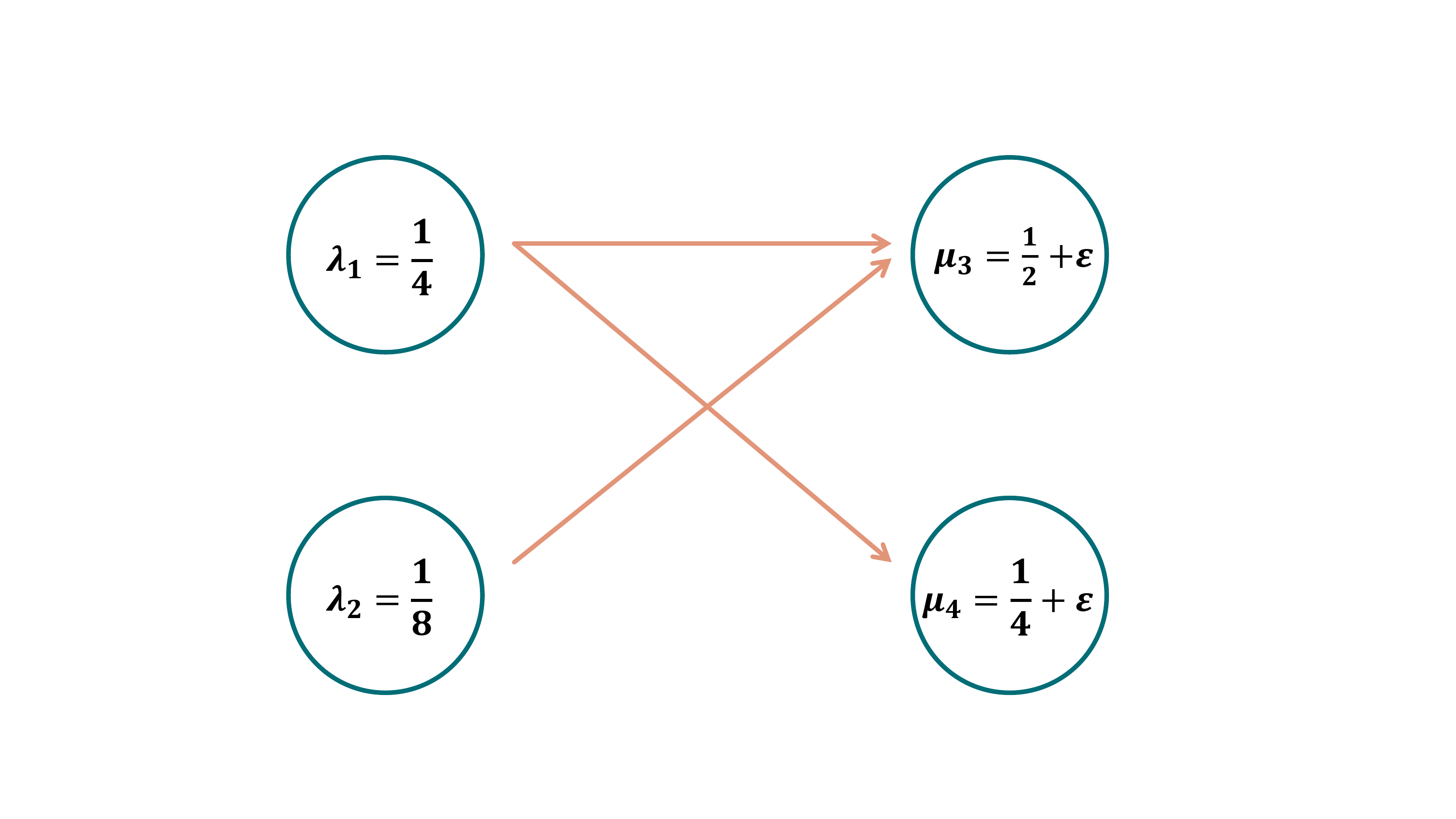}
    \caption{A queueing system with two queues and two servers. $\epsilon$ is sufficiently small.}
    \label{fig:B-incomplete}
\end{figure}
\begin{example}
The queueing system is shown in Figure \ref{fig:B-incomplete}. 
It is easy to verify that it satisfies Assumption~\ref{assump:IB}, which is for any $\boldsymbol{\alpha}\in \{0,1\}^2$, there is a matching matrix $M$, such that $\frac 1 2 \boldsymbol{\alpha}^\top M \processes >(1+\epsilon) \boldsymbol{\alpha}^\top \arrivals$.
For example, when $\bm{\alpha}=[1,1]^\top$, matching matrix 
\begin{equation*}
  \matchingmatrix=
  \left[
  \begin{array}{cc}
0 & 1\\
1 & 0
  \end{array}
  \right]
\end{equation*}
, then $\frac{3}{8}+\epsilon=\frac 1 2 \boldsymbol{\alpha}^\top M \processes >(1+\epsilon) \boldsymbol{\alpha}^\top \arrivals=\frac{3}{8}+\frac{3}{8}\epsilon$.

However, when $\bm{\alpha} \notin \{0,1\}^2$, for example $\bm{\alpha}=[1,\frac{1}{2}]^\top$, 
matching matrix 
\begin{equation*}
  \matchingmatrix=
  \left[
  \begin{array}{cc}
0 & 1\\
1 & 0
  \end{array}
  \right]
\end{equation*}
is the one that maximizes $\boldsymbol{\alpha}^\top M \processes$.
we can see $ \boldsymbol{\alpha}^\top M \processes= \frac{1}{2}(1+3\epsilon)$ and  $\boldsymbol{\alpha}^\top \arrivals= \frac{5}{16}$. Then, $\frac{1}{2}\boldsymbol{\alpha}^\top M \processes < \boldsymbol{\alpha}^\top \arrivals$.
\end{example}

\section{Missing Proofs from Section~\ref{sec:IB-decentral}}
\label{sec:app-IB-decentral}

For any queue~$i$ and $s \in [w]$, during time window $[\ell \cdot w, (\ell+1) \cdot w-1]$, let $G_{i,s}$ be the random variable that the time interval window between arrival of the $(s-1)$-th packet and arrival the $s$-th packet in queue~$i$. 
Then $G_{i,s}$ obeys the geometric distribution with parameter~$\arrival_i$. During the time window, when queue $i$ cleared the $s$-th packet, then its age decreases by $G_{i,s}$ without counting aging since time passed by.
Recall that for server~$j$, $C_t^j$ is the random variable indicating whether the server can successfully process a packet if there is a request at time step $t$.
Each $C_t^j$ is from a Bernoulli distribution with parameter~$\process_j$.

\begin{definition}
  \label{def:eventA}
  For fixed $\eps, \eps_1 > 0$, event $A$ happens when the following inequalities are all true:
  \begin{align}
    &k- \epsilon_1 w  \leq \sum_{s=1}^{k} \arrival_i G_{i,s} \leq k+\epsilon_1 w, \quad  \forall i \in [n] , \forall k \in [w] \label{eq:concentrate on arrival 1}\\
    &\sum_{t = \ell \cdot w}^{(\ell + 1)\cdot w - 1} C_t^j  \geq (1-\epsilon)\process_j w, \quad  \forall j \in [m] \label{eq:concentrate on process 1}\\
    &\sum_{t = \ell \cdot w}^{(\ell + 1)\cdot w - 1} \reg_i(w,(\ell + 1)\cdot w)+n \leq \frac{\epsilon \arrival_{(n)}w}{2}  \label{eq:small no regret 1}
\end{align}
\end{definition}

% Denote by $\arrival_{(n)}$ the smallest arrival rate in $\boldsymbol{\arrival}$. 
\begin{claim}
  \label{claim:eventA}
  Under Assumption~\ref{assump:IB}, the parameters in Definition~\ref{def:eventA} and the no-regret strategies (Definition~\ref{def:no-regret}) can be chosen so that event $A$ happens with probability at least $1-\frac{\beta}{64}$.
\end{claim}

\begin{proof}
The deviation of $\max_k \sum_{s=1}^{k} \arrival_i G_{i,s}^t$ and $\sum_{t=1}^{w} S_j^t$ from their expectations has exponential tail bounds by Corollary~6.7 in \citep{gaitonde2020stability} and Chernoff bound, respectively.
By setting $\delta$ in Definition~\ref{def:no-regret} to be $\frac{\beta}{128n}$, $\epsilon=\frac{\beta}{8}$, $\epsilon_1=\frac{\epsilon \arrival_{(n)}}{2n}$, and for large enough $w$, we can make \eqref{eq:concentrate on arrival 1} and \eqref{eq:concentrate on process 1} hold with probability at least $1-\frac{\beta}{256}$ each, and \eqref{eq:small no regret 1} holds with probability at least $1-\frac{\beta}{128}$. 
Then by the union bound, event~$A$ happens with probability at least $1-\frac{\beta}{64}$.
\end{proof}

\hufu{The language is not quite right.  What we are assuming are what happens in a time window, from $\ell w$ to $(\ell+1)w$, not from time 1 to time~$w$.}

Hereafter we fix $\delta, \eps$ and~$\eps_1$ as in the proof of Claim~\ref{claim:eventA}, and prove the two conditions for $(Z_\ell)_{\ell}$.

\phidecrease*

\begin{proof}
Recall that $G_{i,s}$ is the time interval between the arrival of two adjacent packets. When clearing the $s$-th $\tau$-old packet from queue $i$, the decrease in $T_t^i$ is $G_{i,s}$, which implies the decrease in $\potential_\tau$ is $\arrival_i G_{i,s}$. We don't count the last clearing packets of each queue $i$ since the new age of queue will decrease below $\tau$ and the decrease in $\potential_\tau$ will be less than $G_{i,s}$. Then, we have
\begin{align*}
\Phi_{\tau}(\mathbf{T_{\ell \cdot w}})-\Phi_{\tau}(\bm{\tau}) &\geq  \sum_{i=1}^n \arrival_i \sum_{s=1}^{N_{\tau}^i-1} G_{i,s}\\
&\geq \sum_{i=1}^n (N_{\tau}^i-1- \arrival_i\epsilon_1 w) \\
&\geq N_{\tau}-n-n\epsilon_1 w \\
&\geq \frac{1-\epsilon}{1-\beta} \sum_{ i \in J_{\tau}} \arrival_i w-\sum_{i=1}^n \reg_i(w, (\ell+1) \cdot w)-n-n\epsilon_1 w\\
&\geq \frac{1-2\epsilon}{1-\beta} \sum_{ i \in J_{\tau}} \arrival_i w,
\end{align*} 
where the second inequality uses \eqref{eq:concentrate on arrival 1}, the fourth inequality uses Lemma~\ref{lem:lower-bound-Ntau} and the last inequality is due to \eqref{eq:small no regret 1} and the definition of $\epsilon_1$.

By definition of potential function, we have
\begin{align*}
    \Phi(\mathbf{T_{\ell \cdot w}})-\Phi(\bm{\tau})&=\sum_{\tau=1}^{\infty} \Phi_{\tau}(\mathbf{T_{\ell \cdot w}})-\Phi_{\tau}(\bm{\tau})\\
    &\geq \sum_{\tau=1}^{\infty} \frac{1-2\epsilon}{1-\beta} \sum_{ i \in J_{\tau}} \arrival_i w\\
    &=\frac{1-2\epsilon}{1-\beta}
    \sum_{i=1}^n \arrival_i \tau_i w.
\end{align*}
\end{proof}

\phiincrease*

\begin{proof}
  \hufu{This proof still needs much revision.}
By definition of $\tau_i$, for any $i \in [n]$, $T_{(\ell+1) \cdot w}^i \leq \tau_i + w$, where the equality holds if $\tau_i>0$. 
Let $\bm{w}=\{w,\cdots, w\}$. 
By definition of the potential function, $\Phi(\mathbf{T_{(\ell+1) \cdot w}}) \leq \Phi(\bm{\tau}+\mathbf{w})$. Then,
\begin{align*}
\Phi(\mathbf{T_{(\ell+1) \cdot w}})-\Phi(\bm{\tau}) 
&\leq  \Phi(\bm{\tau}+\mathbf{w})-\Phi(\bm{\tau})\\
&= \frac{1}{2}\sum_{i=1}^n \arrival_i(\tau_{i}+w)(\tau_{i}+w-1)-\frac{1}{2}\sum_{i=1}^n \arrival_i \tau_{i}(\tau_{i}-1)\\
&\leq \sum_{i=1}^n  \lambda_{i} \tau_i w+\frac{1}{2} \sum_{i=1}^n \lambda_{i}w^{2}
\end{align*}
\end{proof}
When $Z_\ell$ is larger than the threshold value, we can relate $ \sum_{i=1}^n  \lambda_{i} \tau_i$ and $ \sum_{i=1}^n \lambda_{i} T_{\ell \cdot w}^i$. 

\begin{lemma}\label{relation between tau and T}
 Under $\emph{event A}$, when $Z_l > \frac{w}{\sqrt{2 \lambda_{(n)}}} \max \left(\frac{8}{\beta}\left(\sum_{i=1}^{n} \lambda_{i}\right), 16 n^{2}\right)$, then $\sum_{i=1}^n  \lambda_{i} \tau_i \geq \frac{1}{2}\sum_{i=1}^n \lambda_{i} T_{l \cdot w}^i$ and $\sum_{i=1}^n \lambda_{i} T_{l \cdot w}^i \geq \frac{8}{\beta}w \sum_{i=1}^n \lambda_{i}$.
\end{lemma}
This Lemma is the combination of Claim 3.1 and Claim 3.2 in \cite{gaitonde2020stability} and we omit its proof.(Though the choice of $\epsilon_1$ is different, it is sufficiently small and will not influence the proof.)

\notA*

\begin{proof}
If Event $A$ does not happen, the worst case is no packets are cleared during the time interval and each age of queue increases by $w$. Then, we have
\begin{align*}
&\Phi(\mathbf{T_{(\ell+1) \cdot w}}) -\Phi(\mathbf{T_{\ell \cdot w}})\\
\leq &\sum_{i=1}^n \frac{1}{2}\arrival_i(T_{\ell\cdot w}^i+w)(T_{\ell\cdot w}^i+w-1)-\frac{1}{2}\arrival_i T_{\ell\cdot w}^i(T_{\ell\cdot w}^i-1)\\
\leq & \sum_{i=1}^n  \lambda_{i}T_{\ell\cdot w}^i w+\frac{1}{2}\sum_{i=1}^n\lambda_{i}w^{2} 
\end{align*}

\end{proof}

\negativedrift*

\begin{proof}
Under \emph{event A}, we calculate the decrease in potential function under :
\begin{align*}
    \Phi(\mathbf{T_{\ell \cdot w}})-\Phi(\mathbf{T_{(\ell+1) \cdot w}})
    & \geq \frac{1-2\epsilon}{1-\beta}\arrival_i \tau_i w - \left(\sum_{i=1}^n  \lambda_{i} \tau_i w+\frac{1}{2}\lambda_{i}w^{2} \right)\\
    &\geq \frac{\beta}{2}\sum_{i=1}^n  \lambda_{i} \tau_i w-\frac{1}{2}\sum_{i=1}^n  \lambda_{i}w^{2}\\
    &\geq \frac{\beta}{4}\sum_{i=1}^n  \lambda_{i} T_{l \cdot w}^i w-\frac{1}{2}\sum_{i=1}^n  \lambda_{i}w^{2}\\
    &\geq \frac{\beta}{8}w\sum_{i=1}^n  \lambda_{i} T_{l \cdot w}^i,
\end{align*}
where the second inequality is due to $\epsilon=\frac{\beta}{8}$ and the last two inequalities is due to Lemma \ref{relation between tau and T}. The probability that \emph{event A} happens is at least $1-\frac{\beta}{64}$. When \emph{event A} does not happen,
\begin{align*}
&\Phi(\mathbf{T_{(\ell+1) \cdot w}})-\Phi(\mathbf{T_{\ell \cdot w}})\\
\leq & \sum_{i=1}^n  \lambda_{i}T_{\ell\cdot w}^i w+\frac{1}{2}\sum_{i=1}^n\lambda_{i}w^{2} \\
\leq &2 \sum_{i=1}^n  \lambda_{i}T_{\ell\cdot w}^i w,
\end{align*}
where the first inequality is due to 
Lemma \ref{lem:not-A}, and the second inequality is due to $\sum_{i=1}^n  \lambda_{i}T_{\ell\cdot w}^i \geq \sum_{i=1}^n\lambda_{i}w$ in Lemma \ref{relation between tau and T}.
Then we calculate the expected change of $\Phi$ when it is large than $b^2$:
\begin{align*}
   \Ex{ \Phi(\mathbf{T_{\ell \cdot w}})-\Phi(\mathbf{T_{(\ell+1) \cdot w}}) \given \Phi(\mathbf{T_{\ell \cdot w}}) > b^2}
    &\geq (1-\frac{\beta}{64}) \cdot w\frac{\beta}{8}\sum_{i=1}^n  \lambda_{i} T_{\ell \cdot w}^i - \frac{\beta}{64} \cdot 2w \sum_{i=1}^n  \lambda_{i}T_{\ell\cdot w}^i\\
    &\geq \frac{1}{2} \cdot \frac{\beta}{8}w\sum_{i=1}^n  \lambda_{i} T_{\ell \cdot w}^i - \frac{\beta}{64} \cdot 2w \sum_{i=1}^n  \lambda_{i}T_{\ell\cdot w}^i\\
    &=\frac{\beta w}{32}\sum_{i=1}^n  \lambda_{i} T_{\ell \cdot w}^i.
\end{align*}
Since for any $a,b,c>0,a-b \geq c \implies \sqrt{a}-\sqrt{b}>\frac{c}{2\sqrt{a}}$, then we have 
\begin{align*}
\Ex{Z_{\ell}-Z_{\ell+1} \given Z_{\ell}>b }&=\sqrt{\Phi_{\ell\cdot w}}-\sqrt{\Phi_{(\ell+1)\cdot w}}\\
&\geq \frac{\beta w \sum_{i=1}^n \lambda_{i} T_{\ell \cdot w}^i}{32 \cdot 2\sqrt{\frac{1}{2}\sum_{i=1}^n \lambda_{i} T_{\ell \cdot w}^i(T_{\ell \cdot w}^i-1)}} \\
&\geq \frac{\sqrt{2 \lambda_{(n)}}\beta w \sum_{i=1}^n \sqrt{\lambda_{i}} T_{\ell \cdot w}^i}{64\sqrt{\sum_{i=1}^n \lambda_{i} T_{\ell \cdot w}^i(T_{\ell \cdot w}^i-1)}}\\
&\geq \frac{\sqrt{2\arrival_{(n)}}\beta w}{64},
\end{align*}
 where the second inequality is due to $\arrival_{(n)}$ is the minimum element of $\arrivals$, and the last inequality is due to $\sum_{i=1}^n a_i \geq \sqrt{\sum_{i=1}^n (a_i)^2}$ when $a_i\geq 0$ for any $i$.
\end{proof}

\thmIBdecentral*

\begin{proof}
$Z_l$ satisfies two conditions in Theorem \ref{pemantale}. Then, for each $a>0$, there is constant $c_a$ such that for any $\ell$, $\Ex{Z_{\ell}^a}\leq c_a$. Then, it means that $\left[\left(\sum_{i=1}^{n} \arrival_{i} T_{\ell \cdot w}^{i}\left(T_{\ell \cdot w}^{i}-1\right)\right)^{a/2}\right] \leq c_a$. 
Then there is a constant $c'_{a}$ such that for any $\ell$, $\Ex{(\sum_{i=1}^{n}T_{\ell \cdot w}^{i})^a} \leq c'_{a}$. Since $\queue_{\ell \cdot w} \leq T_{\ell \cdot w}$ for any $\ell$, then, $\Ex{(\sum_{i=1}^{n}\queue_{\ell \cdot w}^{i})^a}\leq \Ex{(\sum_{i=1}^{n}T_{\ell \cdot w}^{i})^a} \leq c'_{a}$. For any $\ell \cdot w \leq t \leq (\ell+1)\cdot w$, for any $\queue_{t}\leq \queue_{\ell \cdot w}+(n+m)w$. Therefore, for each $a>0$, $t\geq 0$, there is a constant $c''_{a}>0$, such that
\begin{align*}
    \Ex{(\queue_{t})^a}\leq c''_{a}
\end{align*}
Then, we complete the proof of strong stability of the queueing system.
\end{proof}

\section{Missing Proofs from Section~\ref{sec:G-central}}
\label{sec:app-G-central}

\Gcentral*

In this appendix we prove the three statements in Theorem~\ref{central network} are equivalent. 
For convenience, we name the condition in statement 2 the ``edge condition'' % because in the proof, it is the sufficient condition of statement 1. In the same way we call 
and the condition in statement 3 the ``flow condition''.
We first show that the edge condition is sufficient, and that the flow condition is necessary.  
Lastly, we show that the flow condition implies the edge condition, which completes the proof.
% because the sufficiency of statement 2 and the necessity of statement 3 appear more natural, as we show below. Lastly we prove that the ``necessary condition'' implies the ``sufficient condition'', which completes the proof.

\begin{definition}
Given a directed graph $G = (V, E)$, a \emph{path ensemble} $P$ is a collection of vertex-disjoint paths.
The indicator vector of a path ensemble~$P$ is $\indicator_P \in \{0, 1\}^E$ with $\indicator_P (e) = 1$ for $e \in E$ and $0$ otherwise.
We denote by $\paths$ the set of all path ensembles in~$G$.  
\end{definition}

\begin{restatable}{lemma}{pathdecomp}
  \label{path decomposition}
  \label{lem:path-decomp}
Given a feasible solution $\mathbf z$ to the edge condition linear system, there exists a distribution $\mathbf x$ over $\paths$, such that for each edge $(i, j)$ in the network, $\sum_{P \in \paths, P \ni (i, j)} x_{P} = z_{ij}$.
% \end{lemma}
\end{restatable}

This lemma is a generalization of Birkhoff-von Neumann theorem. 
Its proof is reminiscent of textbook proofs of Dilworthy's theorem via Hall's theorem.

\begin{proof}
  Consider the matrix $Z \in [0, 1]^{(\source \cup \medium) \times (\medium \cup \terminal)}$ defined by $Z_{ij} = z_{ij}$ for all $i$ and~$j$.  
  $Z$ is element-wise nonnegative, and its row sums and column sums are bounded by~1 due to (8) and (9) in the edge condition.
%   constraints \eqref{eq:flow-outgoing1} and~\eqref{eq:flow-incoming1}.
% Then, we obtain a doubly stochastic matrix $R$ because of constraints (4) and (5). 
By Birkhoff-von Neuman theorem, 
% $Z$ is a convex combination of matching matrices. 
% $\matchingmatrix$.
$Z$ can be written as $\sum_{P}x_{P}\matchingmatrix_{P}$, where $x_P \geq 0$, $\sum_{P}x_{P}=1$, and each $\matchingmatrix_P \in \{0, 1\}^{(\source \cup \medium) \times (\medium \cup \terminal)}$ is a matching matrix.
Note that each $\matchingmatrix_P$ corresponds to a path ensemble if we construct for each $\matchingmatrix_P$ a set of edges by including $(i, j)$ if and only if $\matchingmatrix_P(i, j) = 1$.
% For each matching matrix $M_{P}$, we can write its corresponding path ensemble $P$ by the following operation: 
% if $\matchingmatrix_{P}(i,j)=1$, then for edge~$e=(i,j)$, $\indicator_P (e) = 1$;
% otherwise $\indicator_P (e) = 0$. In this way, every matching matrix~$M_{P}$ is corresponding to a path ensemble $P$.
% Since $Z=\sum_{P}x_{P}\matchingmatrix_{P}$, 
Therefore $\mathbf x$ is indeed a distribution over $\paths$, and for any $(i,j)\in E$, 
\begin{align*}
  z_{ij}=\sum_{P} x_{P} \matchingmatrix_{P}(i,j)=\sum_{P} x_{P} \indicator_P (e)=\sum_{P\in \paths, P \ni (i, j)} x_{P}.
\end{align*}
% $\mathbf x$ is a distribution over matching matrices since $x_P \geq 0$, $\sum_{P}x_{P}=1$. Due to one to one mapping between matching matrix and path ensemble, $\mathbf x$ is a distribution over $\paths$. We also have $z_{ij}=\sum_{P \in \paths, P \ni (i, j)} x_{P} $ for each edge $(i, j)\in E$, which completes the proof of lemma.
\end{proof}

\begin{lemma}
  If for a Queue-G model queueing system, the edge condition holds,
%   linear system given by constraints \eqref{eq:source-process1}-\eqref{eq:zij-nonnegative} is feasible, 
  then there exists a centralized policy under which the system is stable.
  \label{lem:2to1}
\end{lemma}
\begin{proof}% [Proof of Theorem~\ref{thm:G-central}, Sufficiency]
% \subparagraph{Sufficiency}
  Given a queueing system for which the linear system in the theorem statement is feasible. 
  Let $\allocs$ be a corresponding distribution over path ensembles guaranteed by Lemma~\ref{lem:path-decomp}.
  Consider the central policy that, at each step, samples a path ensemble $P \in \paths$ with probability $x_P$, and then let each node~$i \in \source \cup \medium$ send its packet (if any) to server~$j$ if $(i,j) \in P$.
  (If no edge goes out from~$i$ in~$P$, the node sends no request during that time step.)
% If the queue at node~$i$ is empty and no packet arrives in that time step, the request is empty.
Since $P$ is a collection of vertex-disjoint paths, no server receives more than one request. 
% We now show that it makes the network stable.

We show that, under this policy, at each node whose queue, at each time step, the probability a new packet joins the queue of the node is strictly smaller than the probability a packet from that node is successfully processed by a next server.
It is then straightforward to see that the queue lengths at each server satisfy the conditions of Theorem~\ref{thm:pemantle}, and so does the total number of packets in the system.
The stability of the system follows.
For any node $i \in \source \cup \medium$, the probability its packet is successfully processed is $\sum_{j: (i, j) \in E} \process_j \sum_{P: P \ni (i, j)} x_P = \sum_{j: (i, j) \in E} z_{ij} \process_j$ by the definition of $\allocs$.
For $i \in \source$, the probability a new packet joins its queue is $\arrival_i$ by definition.
For $i \in \medium$, the probability a new packet joins its queue is at most $\process_i \sum_{i: (i, j) \in E} \sum_{P: P \ni (i, j)} x_P = \process_i \sum_{i: (i, j) \in E} z_{ij}$.
% Constraints \eqref{eq:source-process1} and \eqref{eq:middle-process1} 
(8) and (9) in the edge condition precisely require the arrival rate in each case is strictly smaller than the expected process rate.

\end{proof}

\begin{lemma}
  \label{lem:1to3}
% If a queueing system does not satisfy the edge condition, then it does not follow the `flow condition, the system is not stable.
If a queueing system is stable under some centralized policy, then the flow condition holds.
\end{lemma}
\begin{proof}
Consider a queueing system $(V, E, \arrivals, \processes)$ that is stable under a certain centralized policy.
Under this policy, for $i \in \source$ and time step~$t$, denote by $X_t^i$ the number of packets arriving at queue~$i$ at time step~$t$, and $Y_t^i$ the number of packets that depart from queue~$i$ and leave the queueing system at time step~$t$. 
$(X_t^i)_t$ are therefore i.i.d.\@ Bernoulli random variables.
For time step~$T$, let $A_T$ denote the event that, for all $i \in \source$,  $\sum_{t=1}^{T} X_t^i \geq \arrival_i T+\sqrt{T}$. 
By the central limit theorem, for large enough~$T$, $\Prx{A_T}$ is at least a constant $c_0$ which depends on $\arrivals$ but not on~$T$.
% The following analysis is under the event $A_T$. 
From this point on, we condition on event~$A_T$ happening.

For an edge $(x, y) \in E$, let $z_{xy}^t$ be the expected number of request sent by $x$ to~$y$ under the centralized policy at time step~$t$, and let $z_{xy} = \frac 1 T \sum_{t = 1}^T z_{xy}^t$.
Then clearly, for any~$x \in \source \cup \medium$, $\sum_{y \in \outneighbor(x)} z_{xy} \leq 1$, and for any $y \in \medium \cup \terminal$, $\sum_{x \in \inneighbor(y)} z_{xy} \leq 1$.
For a path $\pi$ from a node $i \in \source$ to a node in~$\terminal$, let $f_{i\pi}^T$ be the expected number of packets sent along this path and successfully processed up to time~$T$, and let $f_{i\pi}$ be $f_{i\pi}^T / T$.
It is easy to see that $f_{i\pi}$ satisfies the flow conservation constraints~\eqref{eq:flow-conservation}.
We also have
% We denote $z_{xy}$ the expectation frequency of non-empty requests from node $x$ to node $y$. Then, $\sum_x z_{xy}\leq 1$ and $\sum_y z_{yx}\leq 1$ by definition. Denote by $f_{i\pi}\cdot T$ the expectation flow on path $\pi$ from queue $i$ up to time $T$. Then, 
\begin{align*}
    \sum_{y\in \outneighbor(x)}\sum_{i\in S_1} \sum_{\pi\in\Pi: \pi\ni (x,y)} \frac{f_{i\pi}}{\process_y} &\leq \sum_{y} \frac{z_{xy}\process_y}{\process_y} =\sum_{y} z_{xy} \leq 1, & \forall x \in \source \cup \medium; 
    \\ 
     \sum_{i\in S_1} \sum_{y \in \inneighbor(x)} \sum_{\pi\in\Pi: \pi \ni (y, x)} f_{i\pi}&\leq \process_x \sum_{y \in \inneighbor(x)} z_{yx} \leq \process_x, & \forall x \in \medium \cup \terminal.
 \end{align*}
 The key observation in both inequalities is that at any step, the expected number of requests sent along an edge $(x, y)$ times $\process_y$ is equal to the expected number of packets from~$x$ that is cleared by~$y$.  
 This relationship holds even with the conditioning on~$A$ and other happenings in the system.
 % To see \eqref{eq:f-outflow}, note that $\sum_{i\in S_1} \sum_{\pi\in\Pi:\pi\ni (x,y)} f_{i\pi} T \leq z_{xy}\process_y T$. 
 % The expectation flow on edge $(x,y)$ calculated by $f_{i\pi}T$ is no more than the expectation of packets cleared by server $y$ sent from queue $x$. The second line is due to $\sum_{\pi:x \ on \ \pi} f_{i\pi}T \leq \process_x \sum_y z_{yx}T$. The expectation flow over node $x$ calculated by $f_{i\pi}T$ is no more than the expectation of packets cleared from queue $x$.

We have shown that $(f_{i\pi})_{i \in \source, \pi \in \paths}$ satisfies the first and second inequalities and also satisfies the last equality in the `flow condition. 
If the linear system in the `flow condition is not feasible, then the third inequality is not satisfied. Therefore, there is a queue $i^*$ such that $\sum_\pi f_{i^*\pi} \leq \arrival_{i^*}$.
This means $\mathbb{E}[Y_T^{i^*} |A_T]\leq \arrival_{i^*} T$. 
Therefore, when event~$A_T$ happens, the expected number of packets from queue~$i^*$ that remain in the system is at least $\sqrt T$.
As $A_T$ occurs with constant probability for all large enough~$T$, this contradicts the system's stability under the centralized policy.  
This shows the necessity of the ``necessary condition.''
% This means that when event $A$ happens, at least $\sqrt T$ 
% In proof of Theorem \ref{thm:IB-central}, we prove that there is a constant $c_0$ such that $\lim_{T\rightarrow \infty} P(A_T)= c_0$.  
% Since the non-negativity of $\queue_T^i$,
% Therefore $\mathbb{E}[\queue_T]\geq P(A_T)\mathbb{E}[\queue_T | A_T] \geq P(A_T)\mathbb{E}[\queue_T^{i_T} | A_T]$. 
% Then $\lim_{T\rightarrow \infty} \mathbb{E}[\queue_T] \geq \lim_{T\rightarrow \infty} P(A_T)\mathbb{E}[\queue_T^i | A_T] \geq c_0 \sqrt{T} \rightarrow \infty$.
% Therefore, we prove necessity of the stability condition.
\end{proof}

\begin{lemma}
  \label{lem:3to2}
A queueing system that satisfies the flow condition must also satisfy the edge condition.
\end{lemma}
\begin{proof}
This proof consists of two parts. Part 1 transforms the ``flow condition'' to the ``edge condition'' with $``="$ instead of $``<"$ for \eqref{eq:middle-process1}, and proves others hold. Part 2 replaces $``="$ with $``<"$ for \eqref{eq:middle-process1}, and proves others still hold.

\textbf{Part 1.} Let $(f_{i \pi})_{i \in \source, \pi \in \Pi}$ be a feasible solution to the flow condition. 
% for any vertex $x$, define the solution variable 
For each edge $(x, y) \in E$, define 
$$z_{xy}^1 := \sum_{i\in S_1} \sum_{\pi\in \Pi : \pi \ni (x,y)} \frac{f_{i\pi}}{\process_y}.$$ 

% Then, for (8) in the edge condition, according to (12) in the flow condition, 
Then by \eqref{eq:flow-source} we have
$$\forall x\in \source\cup \medium, \quad \sum_{y}z^1_{xy} = \sum_{y\in \outneighbor(x)}\sum_{i\in S_1} \sum_{\pi\in \Pi : \pi \ni (x,y)} \frac{f_{i\pi}}{\process_y} \le  1.$$

% For (9) in the edge condition, according to (13) in the `flow condition, 
By \eqref{eq:flow-incoming} we have
$$\forall x\in \medium\cup \terminal, \quad \sum_{y}z^1_{yx} = \sum_{y\in \inneighbor(x)}\sum_{i\in S_1} \sum_{\pi\in\Pi: \pi\ni (y, x)} \frac{f_{i\pi}}{\process_x} 
% = \sum_{i\in S_1} \sum_{\pi\in\Pi: \exists y, (y, x)\in \pi} \frac{f_{i\pi}}{\mu_x} 
\le 1.$$

% For (7) in the edge condition, we have (notice flow in and out are equal)
By the flow conservation constraints \eqref{eq:flow-conservation}, we have
$$
\begin{aligned}
    \forall x\in \medium, 
    \quad
    \sum_{y}z^1_{xy} \process_y &= \sum_{y\in \outneighbor(x)}\sum_{i\in S_1} \sum_{\pi\in\Pi: \pi \ni (x,y)} f_{i\pi}  \\
    % = \sum_{i\in S_1}\sum_{\pi\in\Pi: x\ on\ \pi}{f_{i\pi}} \\
    &= \sum_{y\in \inneighbor(x)}\sum_{i\in S_1}\sum_{\pi\in\Pi: \pi \ni (y, x)} f_{i\pi} = \process_x \sum_{x}z^1_{yx}.
\end{aligned}
$$
Note that, compared with \eqref{eq:middle-process1}, here we have equality instead of a strict inequality.

% For (6) in the edge condition, according the third inequality the flow condition, we have (notice that $\{f_{x\pi}\ |\ \pi\in\Pi\}$ is the subset of $\{f_{i\pi}\ |\ i\in S_1, \pi \in \Pi: x \ on \ \pi\}$)
By \eqref{eq:flow-process}, we have
$$\forall x\in S_1, \quad \sum_{y}{z_{xy}^{1}\mu_y} 
% = \sum_{i\in S_1}\sum_{\pi \in \Pi: x \ on \ \pi}{f_{i\pi}} \ge
= \sum_{\pi\in\Pi}f_{x\pi} > \lambda_x.$$

The non-negativity constraints of the edge condition are obviously satisfied by $(z_{xy}^1)$ by \eqref{eq:flow-zero} of the flow condition.
Thus, the edge condition holds for $(z_{xy}^1)$ except for \eqref{eq:middle-process1}.

\textbf{Part 2.} We use $m$  to denote the number of servers in $S_2 \cup S_3$. 
Let $\gamma = \min_{i\in S_1}\frac{\sum_{\pi\in\Pi}{f_{i\pi}}}{\arrival_i} > 1$. 
Fix a topological sorting of the nodes; for any node $x \notin \source$, let $t(x)$ be the index of~$x$ in this sorting, and for $x \in \source$, let $t(x) = 0$.
% For each node~$x$, let $t(x)$ be the shortest distance from $\source$ to~$x$.  ($t(x) = 0$ for $x \in \source$.)
% Define $t(x)$  the topology sort index of vertex $x$. 
For edge $(x, y)$, % let $z_{xy} = z^1_{xy}$ if $x \in \source$; otherwise 
let $z_{xy} = z^1_{xy} \cdot \gamma^{\frac{t(x) + 1}{m + 1}-1}$. 
We have $z_{xy} \le z^{1}_{xy}$; % which means (8) - (11) of the edge condition still hold. 
therefore \eqref{eq:flow-outgoing1}-\eqref{eq:zij-nonnegative} still hold for $(z_{xy})$.
% Since $z_{ij} > z_{ij}^1 / \gamma $ for $i \in \source$,  also holds for $(z_{xy})$.
We show that \eqref{eq:middle-process1} now holds with strict inequality for~$(z_{xy})$.
By definition of topological sorting, for any node $i \in \medium$ and $j \in \inneighbor(i)$, $t(j) < t(i)$.  
Therefore whenever $\sum_{j \in \inneighbor(i)} z_{ji}^1 \process_i > 0$, we have
\begin{align*}
\sum_{j \in \inneighbor(i)} z_{ji} \process_i \leq \gamma^{-1/(m+1)} \sum_{j \in \outneighbor(i)} z_{ij}\process_j < \sum_{j \in \outneighbor(i)} z_{ij} \process_j.
\end{align*}

Finally we show \eqref{eq:source-process1} still holds:

% For \eqref{eq:middle-process1} in the edge condition, since $z_{xy}$ is positively correlated with $t(x)$, then we know that for any node $y,y'$, where $(y,x)$ and $(x,y')$ are edges, then $z_{yx}$ has a greater decrease than $z_{xy'}$, compared with $z_{yx}^1$ and $z_{xy'}^1$, so (7) holds, i.e., we can replace $``="$ with $``<"$. 

% As for (6), we have 
$$
\begin{aligned} \sum_{y}{z_{xy}\mu_y} &= \sum_{y} z_{xy}^1\cdot \mu_y \cdot \gamma^{\frac{t(x)+1}{m+1}-1} \\ &=\sum_{i}\sum_{\pi\in\Pi: x \ on \ \pi}{f_{i\pi}\cdot \gamma^{\frac{t(x)+1}{m+1}-1}} \\ &\ge \sum_{\pi\in\Pi: x \ on \ \pi}{f_{x\pi}\cdot \gamma^{\frac{t(x)+1}{m+1}-1}} \quad\quad &\text{(because this line is the subset of the previous line)}\\ &> \frac{1}{\gamma}\cdot\sum_{\pi\in\Pi: x\ on\ \pi}f_{x\pi}\\ &= \frac{1}{\gamma}\cdot\sum_{\pi\in\Pi}f_{x\pi} \quad\quad &\text{(because\ } f_{x\pi} = 0 \text{\ if\ } x \ not\ on\ \pi)\\ &\ge \frac{\sum_{\pi\in\Pi}f_{x\pi}}{\sum_{\pi\in\Pi}\left(f_{x\pi} / \lambda_x\right)} = \lambda_x \quad\quad & \left( \text{because\ } \forall i\in S_1, \gamma \le \frac{\sum_{\pi\in\Pi}f_{i\pi}}{\lambda_i}\right) \end{aligned}
$$

% Thus we proved that $z_{xy}$ satisfies the edge condition.
\end{proof}

\Gdual*
According to Farkas' Lemma, we can write dual of the linear system
$\alpha_{i}$ denotes the dual variables for first two equalities \ref{eq:source-process1} and \ref{eq:middle-process1}, $\beta$,$\gamma_{i}$ for the last inequalities \ref{eq:flow-incoming1}. $(\alpha_{i},\beta, \gamma_{i}\geq 0)$. Set $\alpha_{j}=0$ if $j\in\terminal$\\

\begin{align}
    \beta-\sum_{(i,j)\in P}(\alpha_i-\alpha_j)\process_j&\geq 0, \quad \forall P \in \paths \\
    \beta-\sum_{i \in source} \alpha_i \arrival_i &<0
\end{align}
Then, the linear system of centralized stability condition \ref{central network} is feasible if and only if the above linear system is infeasible. The above linear system is infeasible if and only if such condition holds: If for $\forall \mathbf{\alpha}>0$, $\exists$ a disjoint path set $U$, s.t. $\sum_{(i,j)\in U}(\alpha_{i}-\alpha_{j})\process_{j}>\sum_{i\in \source}\alpha_{i}\arrival_{i}$.
Therefore, centralized stability condition \ref{central network} is satisfied if and only if the above condition holds.

\section{Missing Proofs from Section~\ref{sec:G-decentral}}
\label{sec:app-G-decentral}
Let $X_{t}^{i}$ denote the number of packets arriving at node~$i$ at time step~$t$. Then, $X_{t}^{i}$ obeys Bernoulli distribution with parameter $\arrival_i$. Recall that for server $j$, $C_{t}^{j}$ is the random variable indicating whether the server can successfully process a packet if there is a request at time step $t$.  For each $j$, for any $t$, $C_{t}^{j}$ is from Bernoulli distribution with parameter $\process_j$.
\begin{definition}
\label{def:eventB}
For fixed $\epsilon>0$, event $B$ happens when the following three inequalities are true:
\begin{align}
    &\sum_{t = \ell \cdot w}^{(\ell + 1)\cdot w - 1}  X_{t}^{i} \leq (1+\epsilon)w \quad, \forall i \in [n] \label{eq:concentrate on arrival2}\\
    &\sum_{t = \ell \cdot w}^{(\ell + 1)\cdot w - 1}  C_t^j \geq (1-\epsilon)w \quad, \forall j \in [m] \label{eq:concentrate on process2}\\
    &\sum_{i \in \source \cup \medium} Reg_i(w,(\ell + 1)\cdot w ) \leq w \label{eq:small no regret2}
\end{align}
\end{definition}
\begin{claim}
  \label{claim:eventB}
  Under Assumption~\ref{decentralized G stability condition}, the parameters in Definition~\ref{def:eventB} and the no-regret strategies (Definition~\ref{def:no-regret}) can be chosen so that event $B$ happens with probability at least $1-\frac{\beta \process_{(m)}}{32(n+m)}$.
\end{claim}
\begin{proof}
The deviation of $\sum_{t = \ell \cdot w}^{(\ell + 1)\cdot w - 1} X_t^i$ and $\sum_{t = \ell \cdot w}^{(\ell + 1)\cdot w - 1} C_t^j$ from their expectations has exponential tail bounds by Chernoff bound. By setting $\delta$ in Definition~\ref{def:no-regret} to be $\frac{\beta \process_{(m)}}{96(n+m)^2}, \epsilon=\frac{\beta}{8}$,
and for large enough $w$, we can make \eqref{eq:concentrate on arrival2} and \eqref{eq:concentrate on process2} hold with probability at least $1-\frac{\beta \process_{(m)}}{96(n+m)}$ each, and \eqref{eq:small no regret2} holds with probability at least $1-\frac{\beta \process_{(m)}}{96(n+m)} $. 
Then by union bound, event~$B$ happens with probability at least $1-\frac{\beta \process_{(m)}}{32(n+m)} $.
\end{proof}

Hereafter we fix $\delta$ and~$\eps$ as in the proof of Claim~\ref{claim:eventB}, and prove the two conditions for $(Z_\ell)_{\ell}$.
Given $\tau_i$ for each node $i$, by Assumption~\ref{assump:G-decentral},let $\disjointpathset^*$ be the vertex-disjoint path such that 
\begin{align}\label{eq: using G-decnetral}
\frac{1}{2}(1-\beta)\sum_{(i,j)\in \disjointpathset^*}(\tau_{i}-\tau_{j})\process_{j}>\sum_{i\in \source}\tau_{i}\arrival_i
\end{align}
\Gphidecrease*
\begin{proof}

We know that decrease in potential function by clearing packets is equal to utility of all queues. Then we have,
\begin{align*}
\Phi(\mathbf{Q_{\ell \cdot w}})-\Phi(\bm{\tau})
&=\sum_{i \in \source \cup \medium}\sum_{t = \ell \cdot w}^{(\ell + 1)\cdot w - 1} u_t^i \\
&\geq \frac{1}{2}\sum_{(i,j)\in \disjointpathset^*} (\tau_{i}-\tau_{j}-w)(1-\epsilon) \process_j w- \sum_{i \in \source \cup \medium}\reg_{i}(w,(\ell + 1)\cdot w)\\
&\geq \frac{\beta}{4}\sum_{(i,j)\in \disjointpathset^*}(\tau_i-\tau_j)
  \process_j w+\frac{1}{2}(1-\frac{5\beta}{8})  \sum_{(i,j)\in \disjointpathset^*}(\tau_i-\tau_j)
  \process_j w-w^2 \sum_{j \in \medium \cup \terminal}\process_j - w\\
&\geq \frac{\beta}{4} \sum_{(i,j)\in \disjointpathset^*}(\tau_i-\tau_j)
  \process_j w+(1+\frac{\beta}{8}) \sum_{i\in \source}\arrival_{i}\tau_{i} w-w^2 \sum_{j \in \medium \cup \terminal}\process_j- w,
\end{align*}
where the first inequality is due to Lemma \eqref{lem:lower-bound-utility}, the second inequality uses \eqref{eq:small no regret2} and $\epsilon=\frac{\beta}{8}$, and the last inequality uses \eqref{eq: using G-decnetral} and $\frac{1-\frac{5\beta}{8}}{ 1-\beta}> 1+ \frac{\beta}{8}$.
\end{proof}

\Gphiincrease*
\begin{proof}
The increase in potential function is due to packets arrival in queueing system. Under \emph{event B}, by inequality \eqref{eq:concentrate on arrival2}, for each $i$, $\sum_{t= \ell \cdot w}^{(\ell +1) \cdot w-1} X_t^i \leq (1+\epsilon)\arrival_i w$. The number of packets arriving at each node is upper bounded. For each $i$, $Q_{(\ell+1) \cdot w}^i \leq \tau_i +(1+\epsilon)\arrival_i w$. Then, 
\begin{align*}
\Phi(\mathbf{Q_{(\ell+1) \cdot w}})-\Phi(\bm{\tau}) &\leq \sum_{i\in \source}\frac{1}{2}(\tau_{i}+(1+\epsilon)\lambda_{i} w)(\tau_{i}+(1+\epsilon)\lambda_{i} w-1)-\frac{1}{2}\tau_{i}(\tau_{i}-1)\\
&\leq \sum_{i\in \source} (1+\epsilon) \lambda_{i}\tau_{i}w+\frac{1}{2}\sum_{i\in \source}(1+\epsilon)^2 \lambda_{i}^{2}w^{2}\\
& \leq (1+\frac{\beta}{8})\sum_{i\in \source}\lambda_{i}\tau_{i}w+\sum_{i\in \source}\lambda_{i}^{2}w^{2},
\end{align*}
where the last inequality is due to $\epsilon=\frac{\beta}{8}$ and $\frac{1}{2}(1+\epsilon)^2 \leq 1$.
\end{proof}

When $Z_{\ell}$ is larger than threshold value, we can relate $\sum_{(i,j)\in \disjointpathset}(\tau_i-\tau_j)
  \process_j$ to $\sum_{i \in 
 \source \cup \medium} Q_{\ell \cdot w}^i$.
 
\begin{lemma}
\label{lem: lower bound of potential decrease}
When $Z_{\ell}>\frac{8\sqrt{2(n+m)}}{\beta  \process_{(m)}}( \sum_{i\in \source}\lambda_i^2 w^{2}+  \sum_{i\in \medium \cup \terminal} \process_j w^{2} +w)$, we can prove that $\frac{\beta}{4} \sum_{(i,j)\in \disjointpathset^*}(\tau_{i}-\tau_{j})\process_j w -w^{2} \sum_{i\in \source} \lambda_{i}^{2}-w^{2} \sum_{i\in \medium \cup \terminal} \process_j - w> \frac{\beta \process_{(m)} w}{16(n+m)} \sum_{i\in \source \cup \medium} \queue_{\ell \cdot w}^{i}$
\end{lemma}

\begin{proof}
Let $k$ be the node which has the longest queue at time step $\ell \cdot w$. Then, $\frac{1}{2}\queue_{\ell \cdot w}^{k}(\queue_{\ell \cdot w}^{k}-1)>\frac{1}{(n+m)}\Phi_{\ell \cdot w}=\frac{1}{(n+m)} Z_{\ell}^2$. We can see that 
$\queue_{\ell \cdot w}^{k}>\sqrt{\frac{2}{(n+m)}} Z_{\ell}>\frac{16}{\beta  \process_{(m)}}(\sum_{i\in \source}\lambda_i^2 w^2+ \sum_{i\in \medium \cup \terminal} \process_j w^2 + w)$.
Next, we will give a lower bound of $\sum_{(i,j)\in \disjointpathset^*}(\tau_i-\tau_j)\process_j$. Consider a consecutive path $p$ from node $k$ to any terminal node. Note that $\tau_j=0$ if $j$ is a terminal node because it has no queue. Choose all edges $(i,j)$ with $\tau_{i}>\tau_{j}$ along path $p$ and form a disjoint path set $\disjointpathset'$. Then, $\sum_{(i,j)\in U^{'}}(\tau_{i}-\tau_{j})\geq \sum_{(i,j)\in p}(\tau_{i}-\tau_{j})=\tau_k-0=\tau_{k}$. Since $\disjointpathset^*$ is the disjoint path set that maximizes $\sum_{(i,j)\in \disjointpathset}(\tau_i-\tau_j)\process_j$, then we have
\begin{align*}
    \sum_{ij\in \disjointpathset}(\tau_i-\tau_j)\process_j
    &\geq \sum_{ij\in \disjointpathset'}(\tau_i-\tau_j)\process_j
    \geq \process_{(m)}\sum_{ij\in \disjointpathset'}(\tau_i-\tau_j)\\
    &\geq \tau_{k}\process_{(m)}
    \geq (\queue_{\ell \cdot w}^{k}-w) \process_{(m)}\\
    &\geq \frac{1}{2}\queue_{\ell \cdot w}^{k}\process_{(m)},
\end{align*}
where the second inequality is because $\process_{(m)}$ is the smallest element of $\processes$, the fourth inequality is due to $\tau_k \geq \queue_{\ell \cdot w}^{k}-w$ and the last inequality is due to $Z_{\ell}$ larger than its threshold value leading to $\queue_{\ell \cdot w}^{k}\geq 2w$.
Now, we are ready to prove this lemma.
\begin{align*}
&\frac{\beta}{4} \sum_{(i,j)\in U}(\tau_{i}-\tau_{j})\process_j w -\sum_{i\in \source} \lambda_{i}^{2}w^{2}-\sum_{i\in \medium \cup \terminal} \process_j w^2- w\\
\geq &\frac{\beta}{8}\queue_{\ell \cdot w}^{k}w\process_{(m)}-\sum_{i\in \source} \lambda_{i}^{2}w^{2}-\sum_{i\in \medium \cup \terminal} \process_j w^2- w\\
\geq &\frac{\beta}{16}\queue_{\ell \cdot w}^{k}w\process_{(m)}\\
\geq &\frac{\beta \process_{(m)} w}{16(n+m)} \sum_{i\in \source \cup \medium} \queue_{\ell \cdot w}^{i},
\end{align*}
where the second inequality is due to $\frac{\beta}{16}\queue_{l\cdot w}^{k}w\process_{(m)}>\sum_{i\in \source} \lambda_{i}^{2}w^{2}+\sum_{i\in \medium \cup \terminal} \process_j w^2+ w$, the last inequality is due to $\queue_{\ell \cdot w}^{k} \geq \frac{\sum_{i\in \source \cup \medium} \queue_{\ell \cdot w}^{i}}{n+m}$ since node $k$ has the longest queue.
\end{proof}

\notB*
\begin{proof}
When event~$B$ does not happen, the worst case is no packets cleared and there is a packet arriving at each source at every time step. 
Then, we have
\begin{align*}
&\Phi(\mathbf{Q_{(\ell+1) \cdot w}})-\Phi(\mathbf{Q_{\ell \cdot w}})\\
\leq &\frac{1}{2}\sum_{i\in \source}(\queue_{\ell \cdot w}^{i}+w)(\queue_{\ell \cdot w}^{i}+w-1)-\frac{1}{2}\queue_{\ell \cdot w}^{i}(\queue_{\ell \cdot w}^{i}-1)
\leq \sum_{i\in \source} \queue_{\ell \cdot w}^i w +w^2
\end{align*}

\end{proof}

\Gnegativedrift*

\begin{proof}

Under \emph{event B}, we calculate the  decrease in potential function :
\begin{align*}
    \Phi(\mathbf{Q_{\ell \cdot w}})-\Phi(\mathbf{Q_{(\ell+1) \cdot w}})=&\Phi(\mathbf{Q_{\ell \cdot w}})-\Phi(\bm{\tau})+\Phi(\bm{\tau})-\Phi(\mathbf{Q_{(\ell+1) \cdot w}})\\
     \geq &\frac{\beta}{4} \sum_{(i,j)\in \disjointpathset}(\tau_i-\tau_j)
  \process_j w-\sum_{i\in \source}\lambda_{i}^{2} w^{2}-\sum_{j \in \medium \cup \terminal}\process_j w^{2} -w\\
  \geq &\frac{\beta \process_{(m)} w}{16(n+m)} \sum_{i\in \source \cup \medium} \queue_{\ell \cdot w}^{i},
\end{align*}
where the first inequality uses Lemma \ref{lem:phi-decrease-G} and \ref{lem:phi-increase-G} and the second inequality uses Lemma \ref{lem: lower bound of potential decrease}.
\emph{Event B} happens with probability at least $1-\frac{\beta \process_{(m)}}{32(n+m)}$, which implies the probability that \emph{Event B} does not happen is at most $\frac{\beta \process_{(m)}}{32(n+m)}$. According to Lemma \ref{lem:not-B}, when \emph{Event B} does not happen, 
\begin{align*}
&\Phi(\mathbf{Q_{(\ell+1) \cdot w}})-\Phi(\mathbf{Q_{\ell \cdot w}})
\leq \sum_{i\in \source} \queue_{\ell \cdot w}^i w +w^2
\leq \frac{17}{16} w \sum_{i\in \source \in \medium}  \queue_{\ell \cdot w}^i,
\end{align*}
where the second inequality is due to $\sum_{i\in \source \cup \medium} \queue_{\ell \cdot w}^i \geq 16w$ since $Z_{\ell}$ is larger than threshold value.
Then, in expectation, 
\begin{align*}
&\Ex{\Phi(\mathbf{Q_{(\ell+1) \cdot w}})-\Phi(\mathbf{Q_{\ell \cdot w}}) \given \Phi(\mathbf{Q_{\ell \cdot w}})>b^2}\\
\geq &(1-\frac{\beta \process_{(m)}}{32(n+m)})\cdot \frac{\beta \process_{(m)}w}{16(n+m)}  \sum_{i\in \source \cup \medium} \queue_{\ell \cdot w}^{i}-\frac{\beta \process_{(m)} }{32(n+m)} \cdot \frac{17}{16}w \sum_{i\in \source \cup \medium} \queue_{\ell \cdot w}^{i}\\
\geq & \frac{31}{32} \cdot \frac{\beta \process_{(m)}w}{16(n+m)}  \sum_{i\in \source \cup \medium} \queue_{\ell \cdot w}^{i}-\frac{\beta \process_{(m)} w}{32(n+m)}\cdot \frac{17}{16} \sum_{i\in \source \cup \medium} \queue_{\ell \cdot w}^i\\
>&\frac{\beta \process_{(m)}w}{64(n+m)} \sum_{i\in \source \cup \medium} \queue_{\ell \cdot w}^i
\end{align*}
Since $a-b \geq c$, then $\sqrt{a}-\sqrt{b}>\frac{c}{2\sqrt{a}}$. 
\begin{align*}
Z_{\ell}-Z_{\ell+1}&=\sqrt{\Phi(\mathbf{Q_{\ell \cdot w}})}-\sqrt{\Phi(\mathbf{Q_{(\ell+1) \cdot w}})}\\
&\geq \frac{\beta \process_{(m)}w\sum_{i\in \source \cup \medium} \queue_{\ell \cdot w}^{i}}{64(n+m) \cdot 2\sqrt{\frac{1}{2}\sum_{i\in \source \cup \medium} \queue_{\ell \cdot w}^{i}(\queue_{\ell \cdot w}^{i}-1)}} \\
&\geq \frac{\sqrt{2}\beta \process_{(m)}w}{128(n+m)}.
\end{align*}
\end{proof}

After proving that random process $Z_\ell$ satisfies Negative drift condition, we are ready to prove that $Z_\ell$ satisfies Bounded jump condition.

\Gboundedjump*

\begin{proof}

At each time step $t$, for each node $i$, $\queue_{t}^i$ changes at most $1$. For each node $i$, $\queue_{\ell \cdot w}^i$ changes at most $w$. Then, the change in $\Phi$ is at most $w\sum_{i\in \source\cup \medium} \queue_{\ell \cdot w}^i+n w^2$. We know that if $a,b,c\geq 0, a-b\leq c$, then $\sqrt{a}-\sqrt{b}\leq \frac{c}{2\sqrt{b}}$. The change in $Z_\ell$ is at most 
\begin{align*}
\frac{w\sum_{i\in \source \cup \medium} \queue_{\ell \cdot w}^i+n w^2}{2\sqrt{\sum_{i\in \source \cup \medium} \frac{1}{2}\queue_{\ell \cdot w}^{i}(\queue_{\ell \cdot w}^{i}-1)}}
\end{align*}

By Cauchy-Schwarz inequality, we have $\sum_{i\in \source \cup \medium} \queue_{\ell \cdot w}^i \leq \sqrt{(n+m)} \sqrt{\sum_{i\in \source\cup \medium} (\queue_{\ell \cdot w}^i)^2}$. Then, change in $Z_{\ell}$ is at most $\frac{\sqrt{2}}{2}(n+m)w+ nw^2$. The $p$-th moment change in $Z_{\ell}$ is at most $(\frac{\sqrt{2}}{2}(n+m)w+ nw^2)^p$, which is a constant irrelevant to $\ell$. Thus, random process $Z_{\ell}$ satisfies Bounded jump condition.
\end{proof}

\thmGdecentral*
\begin{proof}
$Z_\ell$ satisfies two conditions in Theorem \ref{pemantale}. Then, for each $a>0$, there is a constant $c_a$ such that for any $\ell$, $\Ex{z_{\ell}^a}\leq c_a$, which means that $\Ex{\left(\frac{1}{2}\sum_{i \in \source \cup \medium}  \queue_{\ell \cdot w}^{i}\left(\queue_{\ell \cdot w}^{i}-1\right)\right)^{\frac{a}{2}}} \leq c_a$. 
Then there is a constant $c'_{a}$ such that for any $\ell$, $\Ex{(\queue_{\ell \cdot w})^a}=\Ex{(\sum_{i \in \source \cup \medium} \queue_{\ell \cdot w}^{i})^a} \leq c'_{a}$.  For any $\ell \cdot w \leq t \leq (\ell+1)\cdot w$, for any $\queue_{t}\leq \queue_{\ell \cdot w}+(n+m)w$. Therefore, for each $a>0$, $t\geq 0$, there is a constant $c''_{a}>0$, such that
\begin{align*}
  \Ex{(\queue_{t})^a}\leq c''_{a}.
\end{align*}
Then, the queueing system is strongly stable.
\end{proof}

\section{Missing Proofs from Section~\ref{sec:patient}}
\label{sec:app-patient}

\begin{example}
\label{ex:nash-regret}
Consider a queueing system with two queues with arrival rates $\arrival_1=\arrival_2=0.4$ and two servers with service rates $\process_1=0.4+\epsilon, \process_2=1-\epsilon$, where $\epsilon$ is sufficiently small. 
Both servers can serve both queues. 
Strategy profile $\mathbf{p}$ with $p_{11}=p_{22}=1, p_{12}=p_{21}=0$ forms a Nash equilibrium. 
However, it is not a no-regret learning strategy for queue 1. 
To see this, we first observe that for large enough~$t$, 
$\Prx{\queue_t^1 >0} \approx \frac{\arrival_1}{\process_1}=\frac{0.4}{0.4+\epsilon}$
by viewing queue 1 and server 1 as a birth-and-death chain and computing its stationary distribution. 
The distribution of $\queue_t^1$ converges to its stationary distribution, and there is $t_0$ such that for any $t>t_0$, $\Prx{\queue
_t^1 >0}>\frac{0.4}{0.4+\epsilon}-\epsilon > 1-4\epsilon$. 
% It means that with probability at least $1-4\epsilon$, queue 1 sends packets to server 1 at any time step $t$. 
Consider a time window with large enough length $w$ after $t_0$; with high probability, queue 1 clears fewer than $0.5w$ packets. 
If queue 1 always chooses server 2 during this time window instead, it would have sent packets to server 2 at least $(1-4\epsilon)w$ times. 
Similarly, queue 2 sends packets to server 2 at most $(\frac{\arrival_2}{\process_2}+\epsilon)w=
(\frac{0.4}{1-\epsilon}+\epsilon)w \leq 0.4(1+4\epsilon) w $ times. 
At any time step $t$, server 2 would have chosen packets from queue 1 to serve if queue 2 does not send a packet. 
Then, the number of packets from queue 1 that server~2 would have processed successfully is at least $(1-4\epsilon - 0.4(1+4\epsilon))(1 - \eps)w \geq 0.55w$ with high probability.
The regret of queue~1 is therefore linear in~$t$ by playing~$\strats$.
% successfully cleared at least $(0.6-6\epsilon)(1-\epsilon)w - \sqrt{w} \geq 0.5 w$ packets. 
% By playing $\strats_1$, queue 1 clears at most $(0.4+\epsilon)w+\sqrt{w}$ packets, while queue 1 would have cleared at least $0.5w$ packets if queue 1 always chooses server 2, showing the linear regret of this strategy. 
% Therefore, strategy profile $\mathbf{p}$ is not a no-regret strategy.
\end{example}
\section{Adversarial Model for General Graphs}
\label{sec:adversarial}
\subsection{Adversarial queueing model}
\begin{definition}
It is a discrete-time queueing system with $n$ sources and $m$ servers. For each queue, only a subset of servers can serve its packets. At each time step, an adversary sends 0 or 1 packet to each queue $i$. Let $\theta$ be any sequence of $w$ consecutive time steps, $i$ any queue.
Define $N(\theta,i)$ to be the number of packets queue $i$ receives during time interval $\theta$. 
%Define $M(\theta,j)$ to be the number of packets server $j$ succeeding in processing if it receives a packet at each time step during time interval $\theta$.  
Define a $(w,\arrivals)$ adversary: for every sequence $\theta$ of $a$ consecutive time steps and for each queue $i$,  $N(\theta,i)/w \leq \lambda_{i}$, for each queue $i$. For each server $j$, it can successfully process a packet with probability $\process_j$.
%$M(\theta,j)/a \geq \process_{j}$. 
The rest of model is the same as \emph{Queue-G} Model in section \ref{sec:G-decentral-sol}.
The only difference between \emph{Queue-G} Model and adversarial queueing model is that in \emph{Queue-G}, at each time step, whether a packet arriving at a source node 
%and whether a server has the ability to successfully process a packet 
is an independent Bernoulli random variable; while in adversarial queueing model, it is decided by the adversary.
\end{definition}

\subsection{Centralized condition for stability in general graphs}

\begin{theorem}
Given an adversarial queueing model with $n$ source nodes and $m$~servers, with arrival rates $\arrivals = (\arrival_1, \cdots, \arrival_n)$ and processing rates $\processes = (\process_1, \ldots, \process_m)$, there is a centralized policy under which the system is stable if and only
if for any $ \bm{\alpha}\in \{\mathbb{R}_{+}^{n+m} \given \alpha_{i}=0 \ \text{if}\ i \in \terminal\}$, there is a vertex-disjoint path set $\disjointpathset$ such that 
$$
\sum_{(i,j)\in \disjointpathset}(\alpha_{i}-\alpha_{j})\process_{j}>\sum_{i\in \source}\alpha_{i}\arrival_{i}
$$
\end{theorem}   

\subsection{Decentralized condition for stability in general graphs}
\begin{assumption}\label{decentralized adversary}
There is a $\beta>0$ such that for any $ \bm{\alpha}\in \{\mathbb{R}_{+}^{n+m} \given \alpha_{i}=0 \ \text{if}\ i \in \terminal\}$, there is a vertex-disjoint path set $\disjointpathset$ such that 
$$
\frac{1}{2}(1-\beta)\sum_{(i,j)\in \disjointpathset}(\alpha_{i}-\alpha_{j})\process_{j}>\sum_{i\in \source}\alpha_{i}\arrival_{i}
$$
\end{assumption}

\begin{theorem}
  If an adversarial queueing system satisfies Assumption \ref{decentralized adversary}, and queues use no-regret learning strategies with
  $\delta=\frac{\beta \process_{(m)}}{96(n+m)^2}$, then the system is strongly stable.
\end{theorem}

The proof of this theorem is pretty much similar to the proof of Theorem \ref{thm:decentralized G}. The slightly difference is in Definition \ref{def:eventB}, \eqref{eq:concentrate on arrival2} holds with probability 1 in the \emph{adversarial queueing model}, making Claim \ref{claim:eventB} easier to analyze.

\section{Queueing Systems with Multiple Types and Multiple Layers}
\label{sec:GN2}

\subsection{Model Description}
We consider a multi-type discrete time queueing network, which is represented by a directed acyclic graph $G=(V=\source \cup \medium \cup \terminal,E)$ with arrival or processing rates on the nodes. It is a generalization of \emph{Queue-G} Model. 
A node~$i$ with no incoming edge is a \emph{source}, and has an arrival rate~$\arrival_i$. For each $i$, $\arrival_i \in (0,1)$. We define packets coming from node $i$ as type-$i$ packets.
We denote the set of sources by~$\source$. For each type-$i$ packet, only a subset of servers can serve it.
All the other nodes are \emph{servers}, and each server~$j$ has a processing rate~$\process_j$.
A server with no outgoing edge is a \emph{terminal}.
We denote the set of terminals by~$\terminal$.
The set of non-terminal server nodes is $\medium \coloneqq V - \source - \terminal$. 
For each type-$i$ packet, we use $\paths_i$ to denote the set of feasible paths type-$i$ packets can go through.
Define $\indicator_{(x,y)^{i}}=1$ if edge $(x,y)\in \paths_i$.
For $x \in \source \cup \medium$, we denote by $\outneighbor(x)^{i} \coloneqq \{y \in V: \indicator_{(x,y)^{i}}=1\}$ the set of out-neighbors of~$i$ with type-$i$, and for a server~$i$, we denote by $\inneighbor(x)^{i} \coloneqq \{y \in V: \indicator_{(y,x)^{i}}=1\}$ the set of in-neighbors of~$x$ with type-$i$. 

Denote by $\queue_t^{x,i}$ the number of type-$i$ packets in node $x$ at the beginning of time step $t$. 
For all $i \in V$, $\queue_0^i = 0$.
% If $\queue_t^i>0$, it means that at the beginning of time step $t$, node $i$ has a packet to be served.
At each time step $t$, the following events happen, in two phases:
(I) packet sending: 
each node~$x$ with $\queue_t^{x,i}>0$ %and has a non-empty queue
chooses a server~$y$ from $\outneighbor(x)^{i}$ and sends $y$ a type-$i$ packet in $x$'s queue only if $\queue_t^{x,i} \geq \queue_t^{y,i}$;
(II) packet arrival and processing:
at each source $i \in \source$, a packet arrives with probability~$\arrival_i$;
% , independently of all the other events.
at each server $y \in \medium \cup \terminal$, if multiple packets are sent to~$y$, let $x$ be the node which sends the packet and $i$ be the type of packet, it chooses the packet which maximizes $\queue_t^{x,i}-\queue_t^{y,i}$(breaking ties arbitrarily) to proceed, and succeeds with probability~$\process_y$.

The arrivals of tasks to sources and the successes of processing at servers are mutually independent events.
If a packet is cleared by server $j \in \medium$, it then joins the queue of server~$j$; 
if a packet is cleared by a server in~$\terminal$, it leaves the system. 
A packet not chosen by or not successfully processed by a server goes back to the node that sends it.
It follows that any $x \in \terminal$ and any type-$i$ has $\queue_t^{x,i}=0$ at any time step~$t$. 

\subsection{Centralized Stability Condition}
\begin{theorem}\label{thm:GN-central}
There exists a centralized policy under which the system is stable if and only if the following system is feasible:
\begin{align*}
\arrival_i < \sum_{y} z_{iy}^{i} \process_y & \ , \forall i \in \source
\\
\process_x \sum_y z_{yx}^{i}  < \sum_y z_{xy}^{i} \process_y & \ , \forall x \in \medium, \forall i \in \source \textnormal{ with } \prod_{y \in \inneighbor(i)} z_{yx}^{i} > 0
\\
\sum_i \sum_y z_{xy}^{i}  \leq 1 & \ , \forall x \in \source \cup \medium
\\
\sum_i\sum_y z_{yx}^{i}  \leq 1 & \ , \forall x \in \medium \cup \terminal
\\
z_{xy}^{i}=0 & \ , \forall (x,y) \text{not on one of paths for packet} \ i
\end{align*}
\end{theorem}

\begin{definition}
A vertex-disjoint path with type (one such path need not start from a source or end at a terminal) is a collection of edges with type. We use $(x,y,i)$ to represent edge $(x,y)$ with type-$i$ and require that node $y$ can serve type-$i$ packets and node $x$ can serve type-$i$ packets if node $x$ is a server or node $x$ is source $i$. Any two edges in the path can't have the same head or the same tail. 
\end{definition}
Before moving on to decentralized Queue-G models, we derive the dual form of the conditions in Theorem~\ref{thm:GN-central}.
It turns out that the dual form plays a central role in our analysis of the systems' stability under no-regret policies.  
\begin{lemma}
  \label{lem:GN-dual}
  Given a Queue-G model $(V, E, \arrivals, \processes)$ with $n$ sources and $m$ servers, the following two conditions are equivalent:
\begin{enumerate}[(1)]

  \item The linear system of the second statement in Theorem \ref{thm:GN-central} is feasible.
  \item For any $\bm{\alpha}\in \{\mathbb{R}_{+}^{n(n+m)} \given \alpha_{x}^i=0 \ \text{if}\ x \in \terminal \ \text{or if} \ x \in \source \ \text{and} \ x\neq i \}$, there is a disjoint path set with type $ U$ such that $\sum_{(x,y,i)\in U}(\alpha_{x}^{i}-\alpha_{y}^{i})\process_{y}>\sum_{x\in \source}\alpha_{x}^{x}\arrival_{x}$.
\end{enumerate}
\end{lemma}
We omit the proof of Theorem~\ref{thm:GN-central} and Lemma \ref{lem:GN-dual} because they are similar to Theorem~\ref{thm:G-central} and Lemma \ref{lem:G-dual}.

% $\alpha_{x}^{i}$ denotes the dual variables for first two equalities, $\beta_{i}$,$\gamma_{i}$ for the last two inequalities $(\alpha_{x}^{i},\beta_{i}, \gamma_{i}\geq 0)$. Set $\alpha_{j}=0$ if $j\in\terminal$\\the above linear system is infeasible if and only if such condition holds: 

\subsection{Decentralized Stability Condition}
\begin{assumption}\label{decentralized GN stability condition}
There is a $\beta>0$ such that for any $\bm{\alpha}>0$, there is a vertex-disjoint path set with type $\disjointpathset$ such that 
$$
\frac{1}{2}(1-\beta)\sum_{(x,y,i)\in U}(\alpha_{x}^{i}-\alpha_{y}^{i})\process_{y}>\sum_{x\in \source}\alpha_{x}^{x}\arrival_{x}
$$
\end{assumption}
The definition of regret and no regret learning strategy is the same as in \textbf{subsection} \ref{stability and no regret}.

\begin{theorem}\label{thm:decentralized GN}
If Assumption \ref{decentralized GN stability condition} holds, and nodes use no-regret learning strategies with $\delta=\frac{\beta \process_{(m)}}{96(n+m)^2}$, then the queueing system is strongly stable.
\end{theorem}

We define the following potential functions:
\begin{align*}
\Phi\left(\mathbf{\queue}_{t}\right) \coloneqq \frac{1}{2} \sum_{x\in \source \cup \medium} \sum_{i \in \source} \queue_t^{x,i}(\queue_t^{x,i}-1) .
\end{align*}

The main difference between  proof of this theorem and Theorem \ref{thm:decentralized G} is how to prove the following lemma similar to Lemma \ref{lem:lower-bound-utility}. We present the proof of the following lemma and omit the rest of proof.

\begin{lemma}
  For any $\eps > 0$, if \ $\sum_{t = \ell \cdot w}^{(\ell + 1)\cdot w - 1} C_t^y \geq (1 - \eps) \process_y w$ for each~$y$, then for any disjoint path set with type $\disjointpathset$, $\sum_{x \in \source \cup \medium}\sum_{t = \ell \cdot w}^{(\ell + 1)\cdot w - 1} u_t^i \geq \frac{1}{2}  \sum_{(x,y,i)\in \disjointpathset}(\tau_x^i-\tau_y^i -w)(1-\eps)
  \process_y w- \sum_{x \in \source \cup \medium}\reg_{x}(w,(\ell + 1)\cdot w)$.
\end{lemma}

\begin{proof}
Let $\tau_{x}^{i}$ be the number of type-$i$ packets in node $x$ at time step $(\ell+1) \cdot w$ which arrived at node $x$ before time step $\ell \cdot w$. Then, for any $ t\in [\ell \cdot w, (\ell+1)\cdot w], \tau_{x}^{i} \leq \queue_t^{x,i} \leq \tau_{x}^{i}+w$, since each node can only receive and clear packet one at a time step.
Recall that $C_t^y$ is the indicator variable that at time step $t$ if server $y$ can clear a packet successfully. Let $X_{yt}^{s}$ denote the indicator variable that at time step $t$, a packet was sent to server $y$ which generated utility more than $s$ if it was successfully cleared. Set $s=\tau_{x}^{i}-\tau_{y}^{i}-w$.
Then at time step $t$ when $X_{yt}^{s}=0$, node $x$ sending server $y$ a type-$i$ packet because it has priority.  Over the time window, node $i$'s counterfactual utility could have been at least $\sum_{t = \ell \cdot w}^{(\ell + 1)\cdot w - 1}C_{t}^{y}(1-X_{yt}^{s}(Q_{t}^{x,i}-Q_{t}^{y,i}))$. Let $u_t^x$ denote the utility of node $x$ at time step $t$. Let $v_t^y$ denote the utility caused by server $y$ if server $y$ successfully clears a packet at time step $t$. By definition of regret, we have
\begin{align*}
\sum_{t = \ell \cdot w}^{(\ell + 1)\cdot w - 1} u_t^x 
&\geq \sum_{t = \ell \cdot w}^{(\ell + 1)\cdot w - 1} C_{t}^{y}(1-X_{yt}^{s}(Q_{t}^{x,i}-Q_{t}^{y,i}))-\reg_{x}(w,(\ell + 1)\cdot w)  \\
&\geq \sum_{t = \ell \cdot w}^{(\ell + 1)\cdot w - 1} C_{t}^{j}(1-X_{yt}^{s}(\tau^{x,i}-Q_{t}^{y,i}))-\reg_{x}(w,(\ell + 1)\cdot w),
\end{align*}
where the second inequality is due to $Q_{t}^{x,i}\geq \tau_x^i$. On the other hand, whenever $X_{yt}^{s}C_{t}^{y}=1$, it means that server $y$ clears a packet generating utility at least $s$ at time step $t$. Then $v_t^j$ is at least $s$. Over the time window, $\sum_{t = \ell \cdot w}^{(\ell + 1)\cdot w - 1} v_t^j \geq \sum_{t = \ell \cdot w}^{(\ell + 1)\cdot w - 1} C_{t}^{y}X_{yt}^{s}(\tau_x^i-\tau_y^i-w)$.
For each pair $(x,y,i)\in \disjointpathset$, 
\begin{align}\label{eq: lower bound utility of a pair with type}
\sum_{t = \ell \cdot w}^{(\ell + 1)\cdot w - 1} u_x(t)+\sum_{t = \ell \cdot w}^{(\ell + 1)\cdot w - 1} v_t^y
\geq \sum_{t = \ell \cdot w}^{(\ell + 1)\cdot w - 1}(\tau_{x}^i-\tau_{y}^{i}-w)S_t^y-\reg_{x}(w,(\ell + 1)\cdot w),
\end{align}
where the last inequality is due to $\queue_{t}^{y,i}\leq \tau_{y}^{i}+w$.

% Now we are ready to apply Assumption \ref{decentralized GN stability condition}. Let $\alpha_x^i= \tau_x^i$ for any $x,i$.  By Assumption \ref{decentralized GN stability condition}, we can find a vertex-disjoint path with type $\disjointpathset$ such that 
% \begin{align}\label{eq:decentralized-GN condition}
%      \frac{1}{2}(1-\beta)\sum_{(x,y,i)\in U}(\alpha_{x}^{i}-\alpha_{y}^{i})\process_{y}>\sum_{x\in \source}\alpha_{x}^{x}\arrival_{x}.
%  \end{align}
Now, we are ready to give a lower bound of utilities. 
\begin{align*}
2\sum_{x\in \source \cup \medium} \sum_{t = \ell \cdot w}^{(\ell + 1)\cdot w - 1} u_t^x
=&\sum_{x \in \source \cup \medium} \sum_{t = \ell \cdot w}^{(\ell + 1)\cdot w - 1} u_t^x+\sum_{y\in \medium \cup \terminal} \sum_{t = \ell \cdot w}^{(\ell + 1)\cdot w - 1} v_t^y\\
\geq& \sum_{(x,y,i)\in \disjointpathset} (\sum_{t = \ell \cdot w}^{(\ell + 1)\cdot w - 1} u_t^x+\sum_{t = \ell \cdot w}^{(\ell + 1)\cdot w - 1} v_t^y)\\
\geq&  \sum_{(x,y,i)\in \disjointpathset} \sum_{t = \ell \cdot w}^{(\ell + 1)\cdot w - 1}(\tau_x^{i}-\tau_y^{i}-w)S_t^y-\sum_{x \in \source \cup \medium}\reg_{x}(w,(\ell + 1)\cdot w)\\
\geq& \sum_{(x,y,i)\in \disjointpathset} (\tau_x^{i}-\tau_y^{i}-w)(1-\epsilon) \process_j w- \sum_{x \in \source \cup \medium}\reg_{x}(w,(\ell + 1)\cdot w),
% \geq& \frac{\beta}{2} w \sum_{(x,y,i)\in \disjointpathset}(\tau_x^{i}-\tau_y^{i})\process_j+2(1+\epsilon)w \sum_{x\in \source}\arrival_{x}\tau_x^{x}\\
% &-w^2 \sum_{y \in \medium \cup \terminal}\process_y- \sum_{x \in \source \cup \medium}\reg_{x}(w,(\ell + 1)\cdot w),
\end{align*}
where the second inequality uses \eqref{eq: lower bound utility of a pair with type}, and the third inequality is due to $\sum_{t = \ell \cdot w}^{(\ell + 1)\cdot w - 1}S_t^y \geq (1-\epsilon)\process_y w$ which is the condition of the lemma.
% and the last inequality uses \eqref{eq:decentralized-GN condition} and $\epsilon=\frac{\beta}{8}, 1-\frac{\beta}{8}-(1+\frac{\beta}{8})(1-\frac{\beta}{8})\geq \frac{\beta}{2}$.

\end{proof}

\section{A Tighter Analysis of The Queue-CB Model}
\label{sec:app-CB}
We label queues such that $\arrival_1 \geq \cdots \geq  \arrival_n \geq 0$ and label servers such that $\process_1 \geq  \cdots \geq \process_n \geq 0$.
\begin{theorem}
If there exists $\beta>0$ such that for all $k\in[n]$,
$$
\frac{k}{2k-1}(1-\beta)\sum_{i=1}^{k}\process_{i}> \sum_{i=1}^{k}\lambda_{i}
$$ and queues uses no-regret learning strategies with $\delta=\frac{\beta}{128n}$, then the queueing system is strongly stable.
\end{theorem}

The only difference between the proof of this theorem and that of Theorem \ref{thm:IB-decentral} is how to derive Lemma \ref{lem:lower-bound-Ntau}, where we use the condition of this theorem. In the following, we restate and prove this lemma and omit the rest of proof.

\begin{lemma}
  For any $\tau > 0$ and $\eps > 0$, if \ $\sum_{t = \ell \cdot w}^{(\ell + 1)\cdot w - 1} C_t^j \geq (1 - \eps) \process_j w$ for each~$j$, then $N_{\tau} \geq \frac {1 - \eps}{1 - \beta} \sum_{i \in J_\tau} \arrival_i w - \sum_{i = 1}^n \reg_i(w, (\ell+1) \cdot w)$.
\end{lemma}
\begin{proof}
Recall that $N_{\tau}$ is the number of $\tau$-old packets cleared during the time window. Let $k=|J_{\tau}|,l=\min \{k,m\}$. Recall that $X_{jt}^{\tau}$ is the indicator variable for the event that \emph{some queue} (which may not be $i$) sends a $\tau$-old packet to server $j$ at time step $t$.
Then at any time step~$t$ when $X_{jt}^\tau = 0$, server~$i$'s packet would have been picked up by server~$j$ had $i$ sent a request, because no other packet sent to~$j$ is $\tau$-old, so the packet from~$i$ has priority.
Recall that $C_t^j$ is the indicator variable for server~$j$ succeeding in processing a packet if it picks one up.
So queue~$i$ would have gained utility $1$ at time~$t$ by sending a request to~$j$ if $C_t^j = 1$ and $X_{jt}^\tau = 0$. Besides, queue~$i$ would have gained utility $1$ at time~$t$ cleared a packet successfully by server $j$ on the real sample path.
Let $a_{ij}$ denote the number of $\tau$-old packets from queue $i$ cleared by server $j$, we have
\begin{align}
    \sum_{j\in [m]}a_{ij}\geq a_{ij}+ \sum_{t= \ell \cdot w}^{(\ell+1) \cdot w -1}C_{t}^{j}(1-X_{jt}^{\tau})-\reg_{i}(w,(\ell+1) \cdot w) ,
\end{align}
since for any server $j$, a queue would succeed on a time step \emph{when this queue cleared a packet successfully by server $j$ on the real sample path} and when no $\tau$-old packets were sent there.
Summing over queue $i\in J_{\tau}$ , server $j \in [l]$, we have
\begin{align*}
l\sum_{i\in J_{\tau}}\sum_{j\in [m]}a_{ij} \geq \sum_{i\in J_{\tau}}\sum_{j\in [l]} a_{ij}+k\sum_{j\in [l]}\sum_{t= \ell \cdot w}^{(\ell+1) \cdot w -1}C_{t}^{j}(1-X_{jt}^{\tau})-l \sum_{i \in [n]}\reg_{i}(w,(\ell+1) \cdot w)
\end{align*}
Rearranging terms on two sides, we have
\begin{equation}\label{eq:packet number}
\begin{aligned}
&l\sum_{i\in J_{\tau}}\sum_{j\in [m]}a_{ij}+k\sum_{j\in [l]}\sum_{t= \ell \cdot w}^{(\ell+1) \cdot w -1}C_{t}^{j}X_{jt}^{\tau}-\sum_{i\in J_{\tau}}\sum_{j\in [l]} a_{ij}\\
\geq& k\sum_{j\in [l]}\sum_{t= \ell \cdot w}^{(\ell+1) \cdot w -1}C_{t}^{j}-l \sum_{i \in [n]} \reg_{i}(w,(\ell+1) \cdot w)\\
\geq & k (1-\epsilon)  \sum_{j\in [l]}\process_j w-l \sum_{i \in [n]} \reg_{i}(w,(\ell+1) \cdot w),
\end{aligned}
\end{equation}
where the second inequality is due to the condition of this lemma: $\sum_{t = \ell \cdot w}^{(\ell + 1)\cdot w - 1} C_t^j \geq (1 - \eps) \process_j w$ for each~$j$.
Here are some facts based on definitions,
\begin{align}\label{eq:packets equal}
N_{\tau}=\sum_{i\in [n]}\sum_{j\in [m]}a_{ij},
\sum_{j\in [l]}\sum_{t= \ell \cdot w}^{(\ell+1) \cdot w -1}C_{t}^{j}X_{jt}^{\tau}=\sum_{i\in [n]}\sum_{j\in [l]}a_{ij},
\end{align}
where the left and right sides of the second equality represents number of $\tau$-old packets cleared by these $l$ servers, since server $j$ clears a $\tau$-old packet if and only if  $C_{t}^{j}X_{jt}^{\tau}=1$. Then, we give a lower bound of $N_{\tau}$:
\begin{align*}
    &(2k-1)N_{\tau}\\
= & (k-1)\sum_{i\in [n]}\sum_{j\in [m]}a_{ij}+k \sum_{i\in [n]} \sum_{j\in [m]}a_{ij}\\
\geq &(l-1)\sum_{i\in J_{\tau}}\sum_{j\in [l]}a_{ij}+l\sum_{i\in J_{\tau}}\sum_{j\in [m] \backslash [l]}a_{ij}+k\sum_{i\in [n]}\sum_{j\in [l]}a_{ij}\\
=&l\sum_{i\in J_{\tau}}\sum_{j\in [m]}a_{ij}+k\sum_{j\in [l]}\sum_{t= \ell \cdot w}^{(\ell+1) \cdot w -1}C_{t}^{j}X_{jt}^{\tau}-\sum_{i\in J_{\tau}}\sum_{j\in [l]} a_{ij} \\
\geq & k\sum_{j\in [l]}\sum_{t= \ell \cdot w}^{(\ell+1) \cdot w -1}C_{t}^{j}-\sum_{i \in [n]} \reg_{i}(w,(\ell+1) \cdot w)\\
\geq & k (1-\epsilon) w \sum_{j\in [l]}\process_j -\sum_{i \in [n]} \reg_{i}(w,(\ell+1) \cdot w),
\end{align*}
where the first and second equality uses \eqref{eq:packets equal}, the first inequality is due to $ l \leq k$, the second inequality uses \eqref{eq:packet number} and the last inequality uses the condition of this lemma.
Using the condition of Theorem, we have
\begin{align*}
    &\frac{k}{2k-1} \sum_{j\in [l]} \process_{j} 
    \geq \frac{l}{2l-1}\sum_{j\in [l]} \process_{j} \\
    > &\frac{1}{1-\beta}\sum_{i\in [l]} \arrival_{i} 
    \geq \frac{1}{1-\beta}\sum_{i\in J_{\tau}} \arrival_{i}.
\end{align*}
Therefore,
\begin{align*}
    N_{\tau}&\geq \frac{k}{2k-1} (1-\epsilon)\sum_{j\in [l]} \process_{j} w-\sum_{i \in [n]} \reg_{i}(w,(\ell+1) \cdot w)\\
    &> \frac{1-\epsilon}{1-\beta} \sum_{i \in J_{\tau}} \arrival_i w-\sum_{i \in [n]} \reg_{i}(w,(\ell+1) \cdot w).
\end{align*}
\end{proof}

\end{document}